\documentclass[aps,prx,letterpaper,reprint,onecolumn,floatfix,superscriptaddress,notitlepage,usenames,dvipsnames,svgnames,x11names,table,nofootinbib,longbibliography]{revtex4-2}
\pdfoutput=1
\usepackage[utf8]{inputenc}
\usepackage[T1]{fontenc}
\usepackage{graphicx}
\usepackage{grffile}
\usepackage{wrapfig}
\usepackage{rotating}
\usepackage[normalem]{ulem}
\usepackage{amsmath}
\usepackage{textcomp}
\usepackage{amssymb}
\usepackage{pifont}
\newcommand{\cmark}{\ding{51}}
\newcommand{\xmark}{\ding{55}}

\usepackage{capt-of}
\usepackage{verbatim}

\usepackage{parskip}
\usepackage{mathrsfs}
\usepackage[margin= 0.75in]{geometry}
\usepackage[braket, qm]{qcircuit}

\usepackage[english]{babel}
\usepackage{letltxmacro}
\LetLtxMacro{\ORIGselectlanguage}{\selectlanguage}
\makeatletter
\DeclareRobustCommand{\selectlanguage}[1]{%
  \@ifundefined{alias@\string#1}
    {\ORIGselectlanguage{#1}}
    {\begingroup\edef\x{\endgroup
      \noexpand\ORIGselectlanguage{\@nameuse{alias@#1}}}\x}%
}
\newcommand{\definelanguagealias}[2]{%
  \@namedef{alias@#1}{#2}%
}
\makeatother
\definelanguagealias{en}{english}

\usepackage[autostyle,english=american]{csquotes}
\MakeOuterQuote{"}

\usepackage{times}
\usepackage{bbm}
\usepackage{etoolbox}

\usepackage{tcolorbox} 
\tcolorboxenvironment{definition}{colback=RoyalBlue!5!white,colframe=RoyalBlue!75!black}
\tcolorboxenvironment{theorem}{colback=pink!25!white,colframe=pink!100!black}
\tcolorboxenvironment{lemma}{colback=yellow!5!white,colframe=yellow!75!black}
\tcolorboxenvironment{corollary}{colback=Dandelion!5!white,colframe=Dandelion!75!black}
\tcolorboxenvironment{proposition}{colback=Emerald!5!white,colframe=Emerald!75!black}
\tcolorboxenvironment{claim}{colback=RoyalBlue!5!white,colframe=RoyalBlue!75!black}
\tcolorboxenvironment{conjecture}{colback=red!5!white,colframe=red!75!black}
\tcolorboxenvironment{example}{colback=white,colframe=lightgray}
\tcolorboxenvironment{remark}{colback=white,colframe=lightgray}
\tcolorboxenvironment{question}{colback=lightgray!5!white,colframe=lightgray!75!black}

\usepackage{fleqn} 
\usepackage{tikz}
\usepackage{amsthm}
\usepackage{amssymb}
\usetikzlibrary{arrows.meta}
\usepackage{mathtools}

\newcommand{\Tr}{\textnormal{Tr}}

\newcommand{\prlsection}[1]{{\em {#1}.---~}}

\newtheorem{theorem}{Theorem}
\newtheorem{proposition}{Proposition}
\newtheorem{definition}{Definition}
\newtheorem{lemma}{Lemma}
\newtheorem{corollary}{Corollary}

\newcommand{\red}{\textcolor{BrickRed}}
\newcommand{\blue}{\textcolor{RoyalBlue}}
\newcommand{\green}{\textcolor{ForestGreen}}

\usepackage{bm}
\usepackage{dsfont}
\usepackage{ragged2e}
\usepackage[font=scriptsize,labelfont=bf,justification=justified,format=plain]{subcaption}
\usepackage[font=scriptsize,labelfont=bf,justification=justified,format=plain]{caption}
\DeclareCaptionJustification{reallyjustified}{\justifying}
\usepackage{thmtools}
\usepackage{thm-restate}

\definecolor{blueviolet}{rgb}{0.54,0.17,0.89}
\definecolor{mycitecolor}{rgb}{0,0.7,0.1}
\usepackage[pdfusetitle]{hyperref}
\hypersetup{pdflang={English},colorlinks=true,linkcolor=RoyalBlue,citecolor=RoyalBlue,urlcolor=RoyalBlue,breaklinks=true}
\usepackage[nameinlink,capitalize,capitalise]{cleveref}
\usepackage{orcidlink}
\graphicspath{{figures/}}

\begin{document}
\title{Qudit Designs and Where to Find Them}

\author{Namit Anand\orcidlink{0000-0003-4116-4581}}
\thanks{Current address: HPE Quantum, Emergent Machine Intelligence, HPE Labs. Email: namitana@usc.edu (Namit Anand)}
\affiliation{Quantum Artificial Intelligence Laboratory, NASA Ames Research Center, Moffett Field, CA 94035, USA}
\affiliation{KBR, Inc., 601 Jefferson St., Houston, TX 77002, USA}

\author{Jeffrey Marshall}
\affiliation{Quantum Artificial Intelligence Laboratory, NASA Ames Research Center, Moffett Field, CA 94035, USA}
\affiliation{USRA Research Institute for Advanced Computer Science, Mountain View, CA 94043, USA}

\author{Jason Saied}
\affiliation{Quantum Artificial Intelligence Laboratory, NASA Ames Research Center, Moffett Field, CA 94035, USA}

\author{Eleanor Rieffel}
\affiliation{Quantum Artificial Intelligence Laboratory, NASA Ames Research Center, Moffett Field, CA 94035, USA}

\author{Andrea Morello}
\affiliation{School of Electrical Engineering and Telecommunications, UNSW Sydney, Sydney, NSW 2052, Australia}
\affiliation{ARC Centre of Excellence for Quantum Computation and Communication Technology}

\date{\today}

\begin{abstract}
Unitary t-designs are some of the most versatile tools in quantum information theory. Their applications range from randomized benchmarking and shadow tomography, to more fundamental ones such as emulating quantum chaos and establishing exponential separations between classical and quantum query complexity. While unitary designs originating from a group structure, such as the Clifford group, have proven to be incredibly useful for qubit systems, unfortunately, this is no longer true for qudits. In fact, the classification of finite-group representations rules out the existence of unitary 2-designs for arbitrary qudit dimensions. This severely limits the applicability of standard quantum information primitives when it comes to qudit systems.

We overcome these limitations with a three-fold contribution. First, we introduce a general technique to construct families of weighted state t-designs in arbitrary qudit dimensions. These weighted state-designs generalize classical shadow tomography protocol from qubits to qudits. Second, we introduce a Clifford character RB that allows us to benchmark the qudit Clifford group in any dimension, including non-prime-power dimensions. And third, we establish bounds on the quantum circuit complexity of generating approximate unitary-designs from native gates in existing quantum hardware such as high-spin and cavity-QED qudits. Our work further highlights the analogy between spin and optical coherent states by proving that spin-GKP codewords form a state 2-design while spin coherent states do not; in direct analogy with the optical case. This work is structured as a pedagogical and self-contained introduction to unitary designs and their applications to qudit systems.

\end{abstract}

\maketitle
\tableofcontents

\section{Introduction}
\label{sec:intro}

It is difficult to overemphasize the utility of randomness in computation \cite{Hayes_randomness_2001}. In classical computation, (pseudo)randomness shows up in a variety of tasks, the most common perhaps being cryptographic methods, where the security of several protocols fundamentally relies on randomness. In quantum computation, the implications are no less severe. Quantum systems host two different types of randomness, fundamental and statistical. Fundamental being those arising from the structure of quantum theory, such as the noncommutativity of operators, measurement-disturbance relations etc. Statistical (or mixing) being the classical uncertainty on top of the fundamental randomness. This distinction is important when trying to model certain types of randomness, from the well-known random quantum circuit sampling ideas to quantum pseudorandomness \cite{kretschmer2021quantum}. There are various approaches to quantifying how well one can approximate true randomness. For the case of "truly" random quantum processes, often called Haar-random unitaries, this has a well-accepted definition, that of unitary t-designs \cite{gross_evenly_2007}. Unitary t-designs have a simple interpretation: if you're only given access to $t$ copies of a unitary ensemble, they are indistinguishable from Haar-random unitaries. While their cryptographic interpretation is elegant, they also have a beautiful mathematical structure underneath, that of designs. Designs have a rich history in mathematics, often called "cubature rules" \cite{Kuperberg2006}. In particular, spherical designs and their connection to platonic solids have long been appreciated as geometric structures. In quantum information theory, unitary designs have found a multitude of applications, from randomized benchmarking \cite{Emerson2005,Emerson2007,Knill2008,Dankert2009,Magesan2011,Helsen2022} and optimal tomography schemes  \cite{scott_optimizing_2008,roy_unitary_2009} to classical shadow tomography \cite{huang_predicting_2020}, quantum information scrambling \cite{roberts_chaos_2017}, quantum chaos \cite{leone_isospectral_2020}, decoupling methods \cite{szehr2013decoupling}, and even exponential speedups in query complexity \cite{brandao_exponential_2013}. See Ref. \cite{Mele2024introductiontohaar} for an excellent introduction to Haar measure, unitary designs and their applications. An intriguing connection that has emerged in recent years is the relationship between unitary designs and so-called "pseudorandom unitaries" \cite{metger_simple_2024}.

\subsection{Summary of contributions}
This work is written semi-pedagogically, with the intent of bringing together useful results about unitary designs and qudits. One of the challenges here is that quantum information primitives for qudits often come with their own limitations on the dimensionality. For example, certain techniques work for prime-dimensional qudits, while certain techniques work for prime dimensions except two, while others work for prime-power dimensions. Very few work for arbitrary qudit dimensions, which is the focus of this work. Many of these subtleties about the dimensions are summarized in \cref{table:examples-dimensions}.

\cref{sec:intro} introduces unitary designs, complex projective designs or state designs, and their connections with geometric objects such as SIC-POVMs and MUBs. We briefly discuss the notion of approximate designs which is not necessarily tied to a natural group or geometric structure. We focus on exact group designs and representation theory as well as proving some new results about weighted group designs.

\cref{sec:qudits} reviews qudits and various notions of the qudit Clifford group, including the Galois-Clifford group. We discuss the projection lemma that takes a uniform state t-design in dimension \(d\) and projects it onto a weighted t-design in any smaller dimension. We use this to construct families of weighted state 2- and 3-designs in all dimensions. 

\cref{sec:Application to spin and cavity-QED qudits} focuses on experimental platforms for qudits, in particular spin and cavity-QED qudits. We discuss the native gates in both hardware and how to generate unitary designs from these for high-dimensional qudits. We discuss two results that highlight the analogy between spins and bosons, namely that just like optical coherent states, spin-coherent states also cannot form a state 2-design. However, spin-GKP states, like bosonic GKP states, do form a state 2-design. 

\cref{sec:qudit character RB} focuses on the fact that standard Clifford RB which is ubiquitous in qubit systems does not work for qudit systems with a dimension that is not a prime-power. To fix this, we introduce a Clifford character RB that treats all qudit dimensions in essentially the same way. This utilizes the structure of the irreps of the qudit Clifford group.

\cref{sec:fractional designs section} introduces the notion of fractional t-designs which can be used to quantify the closeness to an exact t-design. The cyclic group is used as an example of a "half-design".

\cref{sec:closing} focuses on various applications of these techniques as well as open questions for the future.

\subsection{Unitary designs}
\label{subsec:unitary-designs-intro}
We start by \textit{defining} as ``maximally uniform'' the distribution of unitaries/states that are distributed according to the Haar measure (on the unitary group/state space, respectively). Many of the technical results here are commonplace in the theory of Haar integrals and unitary designs, see e.g., \cite{webb_clifford_2016} for a wonderful introduction. Let \(\mathcal{H} \cong \mathbb{C}^{d}\) be the Hilbert space of our qudit system. Recall that the Haar measure \(\mu\) is the unique, group-invariant, normalized measure over the unitary group, \(U(d)\). That is,
\begin{align}
\quad \mu(U(d)) &= 1; \quad \text{and} \quad \mu (VU) = \mu(UV) = \mu(U) ~~\forall V \in U(d), U \subset U(d).
\end{align}

Suppose we are given an ensemble of unitaries, \(\mathcal{E} = \{ p_{j}, U_{j} \}_{j}\), where \(\sum\limits_{j}^{} p_{j} = 1\) is a probability distribution. Then the \(t\)-fold \textit{twirl} of the ensemble \(\mathcal{E}\) is simply the quantum channel,
\begin{align}
\Phi_{\mathcal{E}}^{(t)} (\cdot) := \sum\limits_{j}^{} p_{j} U_{j}^{\dagger \otimes t} (\cdot) U^{\otimes t}.\label{def:twirl}
\end{align}
The ensemble \(\mathcal{E}\) is a \(t\)-design if and only if the action of the twirled channel is the same as twirling over the entire unitary group with the Haar measure; namely \(\Phi_{\mathcal{E}}^{(t)}(X) = \Phi_{\mathrm{Haar}}^{(t)}(X) ~~\forall X \in \mathcal{L}(\mathcal{H}^{\otimes t})\). Here,
\begin{align}
\Phi^{(t)}_{\mathrm{Haar}} (\cdot) := \int_{\mathrm{Haar}} dU U^{\dagger \otimes t} (\cdot) U^{\otimes t},
\end{align}
where the integral is over the Haar measure on the unitary group. \cref{def:t-designs-equivalent-def} below summarizes various equivalent notions of unitary designs.

\begin{definition}[Equivalent ways of defining unitary \(t\)-designs]\label{def:t-designs-equivalent-def} An ensemble of unitaries \(\mathcal{E} = \{U_{j}\}\) form a unitary \(t\)-design if averaging over them is equivalent to averaging over the full unitary group. This can be expressed in the following equivalent ways:
\begin{enumerate}
\item Averages over \(t\) copies of the unitary and its adjoint are equal: \(\frac{1}{\left| \mathcal{E} \right|} \sum\limits_{j=1}^{\left| \mathcal{E} \right|} U^{\otimes t}_{j} \otimes U^{\dagger \otimes t}_{j} = \int\limits_{\mathrm{Haar}}^{} U^{\otimes t} \otimes U^{\dagger \otimes t} dU\).

\item Averages over \(t\) copies of the unitary channel are equal: \(\frac{1}{\left| \mathcal{E} \right|} \sum\limits_{j=1}^{\left| \mathcal{E} \right|} U^{\otimes t}_{j} X U^{\dagger \otimes t}_{j} = \int\limits_{\mathrm{Haar}}^{} U^{\otimes t} X U^{\dagger \otimes t} dU\), where \(X \in \mathcal{L}(\mathcal{H}^{\otimes t})\).

\item Correlation functions of order \(t\) are equal: \(\frac{1}{\left| \mathcal{E} \right|} \sum\limits_{j=1}^{\left| \mathcal{E} \right|} \operatorname{Tr} \left( \rho U^{\otimes t}_{j} X U^{\dagger \otimes t}_{j} \right) = \int\limits_{\mathrm{Haar}}^{} \operatorname{Tr} \left( \rho U^{\otimes t} X U^{\dagger \otimes t} \right) dU\), for all \(\rho, X \in \mathcal{L}(\mathcal{H}^{\otimes t})\).

\item Polynomials of order \(t\) in the matrix entries \(u_{i,j}\) and \(\overline{u}_{i,j}\) are equal: \(\frac{1}{\left| \mathcal{E} \right|} \sum\limits_{j=1}^{\left| \mathcal{E} \right|} f_{t}(U) = \int\limits_{\mathrm{Haar}}^{} f_{t}(U) dU\) for any such function \(f_{t}\).

\end{enumerate}
\end{definition}

With these definitions in hand, we are ready to introduce our metric. Given an ensemble \(\mathcal{E}\), the \(t\)th \textit{frame potential} is defined as \cite{roberts_chaos_2017},
\begin{align}
F^{(t)}_{\mathcal{E}} := \frac{1}{\left| \mathcal{E} \right|^{2}} \sum\limits_{U,V \in \mathcal{E}}^{} \left| \operatorname{Tr}\left[ U^{\dagger}V \right] \right|^{2t}.
\end{align}
Along with the quantum channel, \(\Phi_{\mathcal{E}}^{(t)}\), let us
also introduce the closely related \textit{moment operator},
\begin{align}
  \Delta_{\mathcal{E},t}  := \sum\limits_{j}^{} p_{j} U_{j}^{\dagger \otimes t} \otimes U_j^{\otimes t}.
\end{align}
The frame potential of \(\mathcal{E}\) is then simply the squared
Hilbert-Schmidt norm of this operator, that is,
\begin{align}
  F_{\mathcal{E}}^{(t)} = \left\Vert \Delta_{\mathcal{E},t} \right\Vert_{2}^{2}.
\end{align}

\begin{lemma}[Bounds on frame potential]\label{lemma:frame potential}
The frame potential for any ensemble of unitaries is lower bounded by the Haar ensemble, \(F^{(t)}_{\mathcal{E}} \geq F^{(t)}_{\mathrm{Haar}}\). Equality is achieved if and only if the ensemble \(\mathcal{E}\) forms a unitary \(t\)-design. Moreover, \(F^{(t)}_{\mathrm{Haar}} = t! \text{ for } t \leq d\).
\end{lemma}

Since \(F^{(t)}_{\mathcal{E}} \geq F^{(t)}_{\mathrm{Haar}}\), their
ratio serves as a natural quantifier for the how well an ensemble \(\mathcal{E}\)
approximates the Haar uniform ensemble. A key property of unitary designs is that a unitary/state \(t\)-design is automatically a \(s\)-design for \(s \leq t, s,t \in \mathbb{N}\). This allows us to talk about the "maximal" order of a \(t\)-design with it being implicit that lower orders exist.

There are various closely related notions of the approximability of
the Haar distribution. We introduce two important definitions
here. The first is based on the diamond norm distance between quantum
channels, \(\Phi_{\mathcal{E}}^{(t)}\) and
\(\Phi_{\mathrm{Haar}}^{(t)}\). As is well-known, the diamond norm
distance quantifies the operational distinguishability of two quantum
channels \cite{watrous_theory_2018}. Therefore, an
\(\epsilon\)-approximate unitary design is defined as,
\begin{align}
  \left\Vert \Phi_{\mathcal{E}}^{(t)} - \Phi_{\mathrm{Haar}}^{(t)} \right\Vert_{\Diamond}^{} \leq \epsilon.
\end{align}

A similar definition which is arguably less operational, but equally
natural mathematically is to
consider the operator norm distance between the two moment
operators. This allows us to introduce a quantum \((\eta,t)\)-tensor product
expander (TPE) as,
\begin{align}
  \left\Vert \Delta_{\mathcal{E},t} - \Delta_{\mathrm{Haar},t} \right\Vert_{\infty}^{} \leq \eta.
\end{align}

Given the fact that all norms are equivalent in finite dimensions, one
can easily relate the two definitions above. A more nontrivial result
shows how to iteratively use a \((\eta,t)\)-TPE to generate an
\(\epsilon\)-approximate \(t\)-design.

\begin{theorem}[Iterating tensor product expanders generates unitary designs \cite{nakata_efficient_2017}]
A \((\eta,t)\) quantum tensor product expander can be iterated \(\ell
\geq\{1 /[\log (1 / \eta)]\} \log \left(d^{t} / \epsilon\right)\)
times to generate an \(\epsilon\)-approximate unitary \(t\)-design.
\end{theorem}

The idea behind this construction is to start by taking a TPE with \(\eta <1\) and iterating it \(k\) times, namely, \(\Delta^k = \underbrace{\Delta \circ \dots \circ \Delta}_{k \text{ times}}\) to obtain a \(\eta^k\)-TPE. And since \(\eta <1\) we have \(\eta^k \ll 1\) and therefore one obtains a better and better approximation by using the TPE as a "seed" to generate larger and larger amounts of randomness. However this iterative property is generically true for the operator norm and not directly for the diamond norm. Similar ideas were used in the seminal work of Brandao \textit{et al.} \cite{brandao_local_2016}.

\subsection{Review of unitary designs literature}
\label{subsec:Review of unitary designs literature}

\renewcommand{\arraystretch}{1.7}
\begin{widetext}
\begin{table}[t!]
\begin{tabular}{||c | c | c | c | c||} 
\hline
Type of designs & Qubits (\(d=2\)) & Qupits ($d=p^k$, \(p\) prime) & Qudits (\(d\neq\) prime-power) & Qunats (CV) (\(d=\infty\)) \\ [0.5ex] 
\hline\hline
State \(1\)-design &  \(\mathbb{B} = \{ | 0 \rangle, | 1 \rangle \}\) & \(\mathbb{B} = \{ | 0 \rangle, | 1 \rangle, \cdots, | p-1 \rangle \}\) & \(\mathbb{B} = \{ | 0 \rangle, | 1 \rangle, \cdots, | d-1 \rangle \}\) & \(\{ | n \rangle_{n \in \mathbb{N}_{0}} \}, \{ | \alpha \rangle_{\alpha \in \mathbb{C}} \}\) \\
 \hline
Unitary \(1\)-group & \(\{ \mathbb{I}, \sigma_{x}, \sigma_{y}, \sigma_{z} \}\) & Qupit Paulis & Qudit Paulis  &  \(D(\alpha) = e^{\alpha a^{\dagger} - \overline{\alpha} a}\)  \\ 
 \hline
State \(2\)-design & stabilizer states &  \begin{tabular}{@{}c@{}} \(k=1\): stabilizer states \\ \(k>1\): \\ orbit of Galois-Cliffords  \cite{heinrich_stabiliser_2021} \end{tabular} &  ? & \red{\xmark} \cite{blume-kohout_curious_2014, iosue_continuous-variable_2024} \\ 
 \hline
Unitary \(2\)-group & Clifford group & \begin{tabular}{@{}c@{}} \(k=1\): Clifford group \\ \(k>1\): \\ Galois-Cliffords \cite{heinrich_stabiliser_2021} \end{tabular} & \red{\xmark} \cite{bannai_unitary_2018} & \red{\xmark} \cite{blume-kohout_curious_2014, iosue_continuous-variable_2024} \\ 
 \hline
Weighted state \(2\)-design & \blue{\cmark} \green{this work} & \blue{\cmark} \green{this work} & \blue{\cmark} \green{this work} & \red{\xmark} \cite{blume-kohout_curious_2014, iosue_continuous-variable_2024} \\
 \hline
Weighted unitary \(2\)-group & \blue{\cmark} \green{this work} & \blue{\cmark} \green{this work} & \red{\xmark} \green{this work} & \red{\xmark} \cite{blume-kohout_curious_2014, iosue_continuous-variable_2024} \\
 \hline
State \(3\)-design & stabilizer states & \red{\xmark} \cite{heinrich_stabiliser_2021} & general \(d\): \red{\xmark} & \red{\xmark} \cite{blume-kohout_curious_2014, iosue_continuous-variable_2024} \\
 \hline
Weighted state \(3\)-design & \blue{\cmark} \green{this work} & \blue{\cmark} \green{this work} & \red{\xmark} \green{this work} & \red{\xmark} \cite{blume-kohout_curious_2014, iosue_continuous-variable_2024} \\
 \hline
Weighted unitary \(3\)-group & \blue{\cmark} & \red{\xmark} \green{this work} & \red{\xmark} \green{this work} & \red{\xmark} \cite{blume-kohout_curious_2014, iosue_continuous-variable_2024} \\
 \hline
\(t \geq 4\) group designs & exist for \(t=4,5\) \cite{gross_evenly_2007} & \red{\xmark} \cite{bannai_unitary_2018} & \red{\xmark} \cite{bannai_unitary_2018} & \red{\xmark} \cite{blume-kohout_curious_2014, iosue_continuous-variable_2024} \\
 \hline
Weighted \(t \geq 4\) groups & \blue{\cmark} & \red{\xmark} \green{this work} & \red{\xmark} \green{this work}  & \red{\xmark} \cite{blume-kohout_curious_2014, iosue_continuous-variable_2024}  \\
 \hline
Rigged \(t \geq 1,2\) designs & \blue{\cmark} \cite{iosue_continuous-variable_2024} & \blue{\cmark} \cite{iosue_continuous-variable_2024}  & ?  & \blue{\cmark} \cite{iosue_continuous-variable_2024}  \\
\hline
\end{tabular}
\caption{Classification of exact state \(t\)-designs and unitary \(t\)-groups for various dimensions. For $d<\infty$, we consider only discrete $t$-designs and $t$-groups. See the extensive discussion in \cref{subsec:Review of unitary designs literature} for more details.}
\label{table:examples-dimensions}
\end{table}
\end{widetext}

We briefly summarize some of the key results that are quite "handy" when working with unitary designs for both qubits and qudit systems. For any dimension \(d\), the generalized Pauli group forms a (minimal) unitary \(1\)-design and so we will rarely discuss this. The situation is much more interesting for \(t=2\)-designs and above. For multiqubits, the \(t\)-design situation can be aptly summarized as "One, Two, Three, Infinity" \footnote{From George Gamow's One Two Three... Infinity: Facts and Speculations of Science.} since the only known examples of unitary designs exist for the case of \(t=1,2,3,\infty\). In the sense that either the group of unitaries form a \(1,2,3\)-design or they form a \(\infty\)-design. The \(1,2,3\)-design comes from the Clifford group. As it turns out, any set of unitaries that forms a \(4\)-design is also a universal generating set (namely these unitaries generate a group containing the special unitary group \(SU(d)\)) and therefore it is automatically a \(\infty\)-design (as in it forms a \(t\)-design for any \(t\in \mathbb{N}\)).

For qudits, the situation is either \(1,2, \infty\) or \(1, \infty\) depending on the dimension. The binary classification is as follows: (i) for prime-power dimensions (namely \(d = p^n\), i.e., \(n\) "qupits"\footnote{Credits to Daniel Lidar for introducing us to this term.}), the Clifford groups form a unitary \(2\)-design. Or one has groups that are universal and hence \(\infty\)-designs. (ii) For non-prime-power dimensions (e.g., \(d = p^a q^b\)), the Clifford group only forms a \(1\)-design. As a result, for arbitrary qudit dimensions, we do not even have an exact unitary \(2\)-design.

We also briefly remark on the situation of \textit{continuous} groups (this is summarized in \cref{thm:no-cts-2-design}). One can show that continuous subgroups of the full unitary group come in only two varieties, they either form \textit{at most} a unitary \(1\)-design or they are \textit{universal} and hence form a unitary \(\infty\)-design. 
It is worth highlighting a particular case of this observation, for $G\subseteq U(d)$ the compact symplectic group. For any initial state $\ket{\psi}$, the orbit $\{U\ket{\psi}: U\in G\}$, with weighting given by the Haar measure on $G$, is a \emph{state} $t$-design for all $t\geq 1$ \cite{west2024random}. However, the symplectic group is \emph{not} universal as a group of unitaries, so \cref{thm:no-cts-2-design} implies that $G$ is \emph{not} a unitary $t$-group for $t\geq 2$. 
This demonstrates the divide between unitary $t$-groups and state $t$-designs: although orbits of a unitary $t$-group are state $t$-designs, state $t$-designs \emph{do not} necessarily give rise to unitary $t$-groups. 

All these results can be summarized as follows: for qubits, the only nontrivial example of a family of unitary $t$-groups is the family of Clifford groups, which form at most a \(3\)-design. For qupits, the Clifford group forms at most a \(2\)-design. And for arbitrary dimensional qudits, one only has nontrivial unitary \(1\)-groups. While sporadic groups exist, we eventually run out of them, and the colloquial statement above holds for \textit{families} of unitary \(t\)-groups.

In \cref{sec:weighted unitary designs}, we observe that weighted versions of unitary $t$-designs do not give ``new'' designs, in the following sense. Let $\mathcal{S}$ be a set of unitaries, and let $G$ be the closed group generated by $\mathcal{S}$. If any weighting of $\mathcal{S}$ gives a weighted unitary $t$-design, then $G$ is in fact an \emph{unweighted} unitary $t$-group (see Theorem~\ref{thm:no weighted designs}). Thus, in Table~\ref{table:examples-dimensions}, the existence of (non-universal) weighted unitary $t$-groups corresponds precisely to the existence of (non-universal) unweighted unitary $t$-groups. 
Further, we see that even \emph{sets} $\mathcal{S}$ of unitaries forming unitary $t$-designs are highly restricted, since the corresponding group $G$ must be an (unweighted) unitary $t$-group. 
Then (as stated in Corollary~\ref{cor:weighted designs universal}), if $t$ and $d$ are such that there is no nontrivial unitary $t$-group in $U(d)$, the search for (exact) unitary $t$-designs is equivalent to the search for universal generating sets. 
In such a context, the more interesting notion is that of \emph{approximate} unitary $t$-design. Approximate designs often correspond to families of bounded-depth unitary circuits. These are generally universal generating sets as well, but the focus of the literature here is on how closely one can approximate a $t$-design using random circuits of the relevant type with a given depth. The state-of-the-art here is that unitary \(2\)-designs on \(n\)-qubits can be obtained in \textit{logarithmic}-depth. More precisely, the depth is \(O(t \text{poly}\left(\log(t)\right) \log(n))\) \cite{schuster_random_2024} for a unitary \(t\)-design on \(n\) qubits, which provides an exponential improvement over the previous estimate of \(O(tn)\) in \cite{hunter-jones_unitary_2019}; which itself was an improvement over the original estimate of \(O(t^2n^{10})\) \cite{brandao_local_2016}.

\subsection{Complex projective designs or state designs}\label{sec:state designs}
One can borrow the notion of uniformity from the unitary group to quantum states. Let the ``maximally uniform'' set of quantum states be the ones that are Haar-uniformly distributed. I.e., the distribution generated by starting from an arbitrary (but fixed) state \(| \psi_{0} \rangle\) and applying Haar-random unitaries to it, i.e., \(\{ | \psi_{U} \rangle := U | \psi_{0} \rangle \}, U \in U(d)\). Note that this distribution is unitarily invariant and hence does not depend on the choice of the fiducial state \(| \psi_{0} \rangle\). Strictly speaking, one is interested in the \textit{projective representation}, namely, we will work with rank-one projectors \(| \psi \rangle \langle  \psi |\) (rays in the Hilbert space) as opposed to vectors in \(\mathcal{H}\). In this sense, it is more accurate to think of these as \textit{complex projective} designs \cite{roy_weighted_2007}.

Then consider \(\mathbb{E}_{\mathrm{Haar}} \left( | \psi \rangle \langle  \psi |  \right)^{\otimes t}\), i.e., the average over \(t\)-fold tensor product of Haar-random states. Note that this is equal to some linear operator in \(\mathcal{L}(\mathcal{H}^{\otimes t})\). For \(t=1\), this is simply the maximally mixed state \(\mathbb{I}/d\) and for \(t=2\) is equal to \(\frac{1}{d(d+1)} \left( \mathbb{I} + \mathbb{S} \right)\), where \(\mathbb{S}\) is the swap operator. For a general \(t\), we have the formula,
\begin{align}
\mathbb{E}_{\mathrm{Haar}} \left( | \psi \rangle \langle  \psi |  \right)^{\otimes t} = \frac{1}{{\binom{d+t-1}{t}}} \Pi_{\mathrm{sym}},
\end{align}
where \(\Pi_{\mathrm{sym}}\) is the projector onto the symmetric subspace. Note, \(\Pi_{\mathrm{sym}} = \frac{1}{t!} \sum\limits_{\pi \in S_{t}}^{} W_{\pi}\) with \(S_{t}\) the permutation group for \(t\) elements and \(W_{\pi}\) is a permutation operator. That is, consider \(\pi \in S_{t}\) such that, \(\pi = \left( \pi(1), \cdots, \pi(t) \right)\) is a permutation of \(t\) objects. Then, \(W_{\pi} | \psi_{1}, \psi_{2}, \cdots, \psi_{t} \rangle = | \psi_{\pi(1)}, \psi_{\pi(2)}, \cdots, \psi_{\pi(t)} \rangle\). The Welch test provides a simple criterion to see if a collection of states forms a \(t\)-design, summarized in the following theorem.

\begin{definition}[Welch test for weighted state designs]
An ensemble of quantum states, \(\mathcal{E} = \{ w_{j}, | \psi_{j} \rangle \}\) forms a weighted \(t\)-design if and only if \cite{roy_weighted_2007}
\begin{align}
\label{eq:welch-test-weighted-design}
\sum\limits_{j,k=1}^{\left| \mathcal{E} \right|} w_{j} w_{k} \left| \left\langle \psi_{j} | \psi_{k} \right\rangle \right|^{2t} = \frac{1}{\binom{d+t-1}{t}}.
\end{align}
\end{definition}
For \(t=1,2\), the RHS gives us \(\frac{1}{d}, \frac{2}{d(d+1)}\), respectively. It is easy to show that tensor product of complex projective \(t\)-designs only form a \(t\)-design for the case of \(t=1\). This follows from the observation that state 1-designs are simply an orthonormal basis and therefore their tensor products from a small space generates a basis for a larger Hilbert space. In contrast, state 2-designs necessitate entanglement, and therefore cannot be obtained by simply tensoring smaller designs.

\begin{lemma}[Bounds on the cardinality of weighted unitary designs]
In \cite[Corollary 20 \& Theorem 21]{roy_unitary_2009}, the following bounds on the size of weighted unitary \(t\)-designs were established,
\begin{align*}
\left| \mathcal{E} \right| \geq D(d,\lceil t / 2\rceil,\lfloor t / 2\rfloor),
\end{align*}
with equality only if \(\mathcal{E}\) is a unitary \(t\)-design. And, for any \(t,d\) there exists a weighted unitary \(t\)-design with cardinality,
\begin{align*}
\left| \mathcal{E} \right| \leq D(d,t,t).
\end{align*}
Here, \(D(d, r, s)\) are the sizes of irreducible representations of \(U(d)\) in \(\mathbb{C}^{d \otimes r} \otimes \mathbb{C}^{* d \otimes s}\). Using Theorem 7 of \cite{roy_unitary_2009} we have that \(D(d, r, s) =\sum_{|\mu|=r-s,|\mu+| \leq r} d_{\mu}{ }^{2}\) where the sum is over nonincreasing, length-\(d\) integer sequences \(\mu\) and \(d_{\mu}=\prod_{1 \leq i<j \leq d} \frac{\mu_{i}-\mu_{j}+j-i}{j-i}\).
\end{lemma}

\subsection{Rigged designs}
For finite-dimensional systems, the "rigged \(t\)-designs" introduced in Ref. \cite{iosue_continuous-variable_2024} can be simply thought of as ``measurement \(t\)-designs''. For infinite-dimensional systems, the use of the rigged Hilbert space becomes necessary but the ideas simply for the finite-dimensional qudit systems we are interested in. Especially if we focus on a "hard cutoff" for bosonic modes.

A set of operators \(\{ M_{j} \}_{j}\) form a POVM if and only if they are positive semidefinite and sum to the identity, namely,
\begin{align}
0 \leq M_{j} \leq \mathbb{I} \text{ and } \sum\limits_{j}^{} M_{j} = \mathbb{I}.
\end{align}
Just as (uniform) state designs can be obtained as the orbit of an arbitrary state under the elements of a (uniform) unitary design, the same applies to measurement designs. Let \(\mathcal{E} = \{ U_{j} \}_{j}\) be a \(t \geq 1\) unitary design. And define \(M_{j}:= U_{j} M U^{\dagger}_{j}\) as the elements of a set \(\mathcal{M} = \{ M_{j} \}_{j}\) with \(M \geq 0\). Then,
\begin{align}
\frac{1}{\left| \mathcal{E} \right|}\sum\limits_{j}^{} M_{j} = \frac{1}{\left| \mathcal{E} \right|}\sum\limits_{j}^{} U_{j} M U^{\dagger}_{j} = \operatorname{Tr}\left[ M \right] \frac{\mathbb{I}}{d},
\end{align}
where we have used the \(1\)-design property of \(\mathcal{E}\). If we renormalize each \(M_{j}\) as
\begin{align*}
M_{j} \mapsto \frac{d}{\left| \mathcal{E} \right| \operatorname{Tr}\left[ M \right]} M_{j}
\end{align*}
then the set \(\mathcal{M}\) forms a POVM. For this renormalization, we assume that \(0 < \operatorname{Tr}\left[ M \right] < \infty\) that is, the initial operator \(M\) is both traceclass and has a nonzero trace. Two trivial examples here include: (i) \(M = \mathbb{I}\) which generates trivial POVM elements (in the sense that these measurements do not extract any information from the state) and (ii) \(M = | \psi \rangle \langle  \psi |\) which is equivalent to using elements of a state design as the elements of a POVM. Here choosing \(\mathcal{E}\) to be a \(1\)-design such as the Pauli group elements would be akin to reducing the POVM to a set of projective measurements (since the elements of \(\mathcal{M}\) are mutually orthonormal). If \(\mathcal{E}\) is \(t \geq 2\)-design then the elements are not orthonormal (every unitary \(2\)-design must have \(> d\) elements) and therefore they form a proper POVM (and not just a projective measurement in disguise).

Therefore, by starting from a single nontrivial positive semidefinite operator \(M\) with \(0 < \operatorname{Tr}\left[ M \right] < \infty\) and ``evolving'' this by elements of a unitary design, we have generated a POVM. Going back to the ``rigged \(t\)-designs'' introduced in Ref. \cite{iosue_continuous-variable_2024}; for finite-dimensional Hilbert spaces we can think of them as ``measurement \(t\)-designs''. Namely, they are POVM elements that can generate the \textit{unnormalized} symmetric projector. Notice that any rigged \(t \geq 1\) design is a rigged \(t=1\)-design. For \(t=1\) this reduces to the set of states \(| \chi \rangle\) to be such that,
\begin{align}
\int\limits_{X}^{} | \chi \rangle \langle  \chi |  d\mu(\chi) \propto \mathbb{I}.
\end{align}
Since the projectors satisfy \(0 \leq | \chi \rangle \langle  \chi | \leq \mathbb{I} \), by renormalizing the measure appropriately, we get a set of (possibly continuous) POVM elements. Moreover, thinking of rigged \(t\)-designs as generalized measurements also allows us to construct them directly. For example, it is easy to see that any state \(t\)-design can be used to generate a rigged \(t\)-design. 

\subsection{SIC POVMs and MUBs}
The original theory of frame potentials originates from quantum states, in particular the study of symmetric, informationally complete (SIC) POVMs. A SIC is a collection of \(d^{2}\) unit vectors in a \(d\)-dimensional Hilbert space, that are ``equiangular,'' in the sense that,
\begin{align}
\left| \langle \psi_{j} |  \psi_{k} \rangle \right|^{2} = \frac{1}{d+1}, \quad 1 \leq j,k \leq d \text{ and } j \neq k.
\end{align}

Recall that a set of operators form a POVM iff \(0 \leq M_{j} \leq \mathbb{I}\) and \(\sum_{j=1}^{m} M_{j} = \mathbb{I}\). Informational completeness is the statement that the POVM elements form a basis for the operator space, i.e., \(\mathrm{span} \{ M_{j} \} = \mathcal{L}(\mathcal{H})\) and as a result, any density matrix can be written as a linear combination of them. The symmetry corresponds to the fact that the Hilbert-Schmidt inner product between the POVM elements is symmetric, \(\operatorname{Tr}\left[ M_{j} M_{k} \right] = c ~~\forall j \neq k\).

 As an example, let us consider the qubit case. Here, the eigenvectors of the Pauli operators,
\begin{align*}
\{ | 0 \rangle \langle  0 | , | 1 \rangle \langle  1 | , | + \rangle \langle  + | , | - \rangle \langle  - | , | +i \rangle \langle  +i | , | -i \rangle \langle  -i |   \}
\end{align*}
with a prefactor of \(\frac{1}{3}\) forms an informationally (overcomplete) POVM. In order to generate a \textit{minimal} IC-POVM, we can instead choose the set,
\begin{align*}
    \{ | 0 \rangle \langle  0 | , | + \rangle \langle  + | , | +i \rangle \langle  +i | , \underbrace{| 1 \rangle \langle  1 | + | - \rangle \langle  - | + | -i \rangle \langle  -i |}_{\text{single POVM}}  \}.
\end{align*}
In contrast, an example of a qubit SIC-POVM is given by the following (nontrivial) set of states,
\begin{align*}
\{ | 0 \rangle, \frac{1}{\sqrt{3}}|0\rangle+\sqrt{\frac{2}{3}}|1\rangle, \frac{1}{\sqrt{3}}|0\rangle+\sqrt{\frac{2}{3}} e^{i \frac{2 \pi}{3}}|1\rangle, \frac{1}{\sqrt{3}}|0\rangle+\sqrt{\frac{2}{3}} e^{i \frac{4 \pi}{3}}|1\rangle  \}.
\end{align*}
These states form a regular tetrahedron inside the qubit Bloch sphere, which is reminiscent of their equiangular structure.

\prlsection{Mutually Unbiased Bases} Let \(\mathcal{H} \cong \mathbb{C}^{d}\) and \(\mathbb{B}, \tilde{\mathbb{B}}\) be two (different) orthonormal bases for \(\mathcal{H}\). Then, \(\mathbb{B}\) is \textit{mutually unbiased} with respect to \(\tilde{\mathbb{B}}\) if and only if
\begin{align}
\left| \langle \psi_{j} | \tilde{\phi}_{k} \rangle \right|^{2} = \frac{1}{d} ~~\forall | \psi_{j} \rangle \in \mathbb{B}, | \tilde{\phi}_{k} \rangle	\in \tilde{\mathbb{B}}.
\end{align}
An immediate consequence of this definition is that given a state \(| \psi_{j} \rangle \in \mathbb{B}\), measurements in the basis \(\tilde{\mathbb{B}}\) yield an \textit{unbiased} estimator of the state. In the case of a single qubit, the computational basis is mutually unbiased w.r.t. the eigenvectors of both Pauli \(\sigma_{x}\) and \(\sigma_{y}\). For any qudit, given an orthonormal basis, one can generate a MUB by applying a (discrete) quantum Fourier transform.

\begin{lemma}[Equivalence of state designs, SIC-POVMs, and MUBs \cite{klappenecker_mutually_2005}] For any \(d\)-dimensional Hilbert space, the following equivalences hold. \\
\begin{enumerate}
    \item SIC-POVMS form a (derandomized) state \(2\)-design, and\\
    \item a collection of \(d+1\) MUBs form a state \(2\)-design.
\end{enumerate}
\end{lemma}

\subsection{Approximate versus exact designs}

A seminal work of Brandao et. al. showed that local random quantum circuits generated in a brickwork fashion from \(2\)-local Haar random unitaries of depth \(\text{poly}\left( t \right)\) form an approximate \(t\)-design \cite{brandao_local_2016}. More recently, it was shown that \cite{haferkamp_efficient_2023} inserting \(O(t^{4} \log^{2}(t) \log(\frac{1}{\epsilon}))\) many non-Clifford gates (e.g., \(T\)-gates) is sufficient to transform a random Clifford circuit into an \(\epsilon\)-approximate \(t\)-design. While there has been an extensive amount of work on the generation of \(t\)-designs in brickwork circuits (also known as random quantum circuits, local random circuits, etc., see Ref. \cite{fisher_random_2023} for a recent review), the structure assumed in these circuits is not well-suited for our purposes. For example, these circuits typically assume the ability to perform \(2\)-qubit Haar random unitaries on any pair of qubits; which amounts to a \(4 \times 4\) random unitary. In contrast, we are interested in say, a qudit of local dimension \(d=16\) (e.g., in this high-spin donor qudit architecture \cite{FernndezdeFuentes2024}) and keeping the same structure amounts to assuming the ability to "easily generate" a \(16 \times 16\) random unitary. Namely, as we increase the local Hilbert space dimension, assuming the ability to perform local random unitaries is a \textit{nontrivial} assumption on the quantum computing architecture. 

As a result, we would instead like to focus on a different approach, one followed in Refs. \cite{nakata_diagonal-unitary_2013,nakata_efficient_2017}, summarized below. Since our primary focus here is on a single qudit, this is a more natural construction for us. E.g., as we will argue below, we would like to perform random diagonal unitaries in two different bases and understand how fast they converge to a \(t\)-design.

It is worth mentioning that in Ref. \cite{nakata_diagonal-unitary_2013} the authors introduce a ``diagonal unitary design,'' namely, unitaries that can approximate the uniform distribution over diagonal unitaries in a fixed basis \(\mathbb{B}\). This is not the usual unitary designs that are discussed in the literature which are aimed at approximating \textit{Haar random} unitaries. In Ref. \cite{nakata_efficient_2017} the authors introduce a method to alternatively apply random unitaries from MUBs to generate \textit{approximate} \(2\)-designs.

\subsection{Unitary designs and representation theory}\label{sec:rep theory}

In this section, we will review the connection between unitary $t$-designs and representation theory. These connections have previously been discussed, for example, in Refs. \cite{gross_evenly_2007,zhu_clifford_2016,mele_introduction_2023} and Ref. \cite{Saied2024} by a subset of the authors. 
We will consider a (finite or compact) subgroup $G$ of the unitary group $U(d)$, and the corresponding ensemble $\mathcal{E}$ whose unitaries are the elements of $G$, with weights $1/|G|$ in the finite case or determined by the Haar measure in the (infinite) compact case. 
For convenience, we refer to the corresponding frame potential as $F_{G}^{(t)}$. 
We briefly review some standard representation-theoretic concepts before stating the result. 

A \emph{representation} of a group $G$ is given by a pair $(\rho, V)$ of a vector space $V$ and a group homomorphism $\rho: G\rightarrow GL(V)$. 
A representation is often referred to by the vector space $V$ alone, especially if the map $\rho$ is understood from context. 
An important example is the \emph{natural representation} of a matrix group $G\subseteq GL(d)$ on $\mathbb{C}^d$ by matrix-vector multiplication; formally, the natural representation corresponds to the pair $(I, \mathbb{C}^d)$ where $I$ is the identity map on $G$. 
We generally refer to $\mathbb{C}^d$ as the natural representation of $G$ without specifying the homomorphism $\rho=I$. 

Two representations $(\rho_1, V_1), (\rho_2, V_2)$ are \emph{equivalent} if there exists a linear map $\phi: V_1\rightarrow V_2$
that is an isomorphism of vector spaces and commutes with the group action: that is, for all $v_1\in V_1$ and $g\in G$, 
    $\phi(\rho_1(g)v_1) = \rho_2(g) \phi(v_2).$ 
This is the relevant notion of isomorphism for representations, and (when the maps $\rho_1, \rho_2$ are understood) we write $V_1\cong V_2$. 

Given a representation $(\rho, V)$ of $G$, a \emph{subrepresentation} is a vector space $W\subseteq V$ such that, for all $g\in G$ and $w\in W$, $\rho(g)w\in W$. 
This is a new representation of $G$, defined by $(\rho', W)$, where 
    $\rho'(g) = \rho(g)|_W.$  
Similarly, given two groups $G_1\subseteq G_2$ and a representation $(\rho, V)$ of the larger group $G_2$, we obtain the \emph{restricted representation} of the smaller group $G_1$ as $(\rho|_{G_1}, V)$, where 
    $\rho|_{G_1}(g) = \rho(g)$
is simply the restriction of $\rho$ to $G_1$. 

Given two representations $(\rho_1, V_1), (\rho_2, V_2)$ of $G$, we may form the \emph{direct sum representation} $(\rho, V)$ where $V = V_1\oplus V_2$ and, for $g\in G$ and $v_i\in V_i$, 
    $\rho(g)(v_1 \oplus v_2) = \rho_1(g)v_1 \oplus \rho_2(g) v_2.$ 
We also have the \emph{tensor product representation}, with $V= V_1\otimes V_2$ and 
    $\rho(g)(v_1\otimes v_2) = (\rho_1(g)v_1)\otimes (\rho_2(g)v_2).$ 

A representation $(\rho, V)$ is \emph{irreducible} if it has no nonzero proper subrepresentations: 
in other words, if $W\subseteq V$ is a subrepresentation, then $W=0$ or $W=V$. 
The term "irreducible representation" is often shortened to \emph{irrep}. 
A representation $(\rho, V)$ of $G$ is \emph{completely reducible} if it decomposes into a direct sum of irreps: $V = W_1\oplus W_2\oplus \cdots \oplus W_k,$ where each $W_i$ is an irrep for $G$. 
Typically we are only concerned with such decompositions up to equivalence, writing $V \cong \bigoplus_i V_i^{\oplus n_i},$ 
where each $V_i$ is irreducible, $V_i\cong V_j$ iff $i=j$, and $n_i$ is called the \emph{multiplicity} of the irrep $V_i$. 
For the groups of interest to us here, specifically finite and compact groups, every finite-dimensional representation $(\rho, V)$ is completely reducible. 

To connect this theory to $t$-designs, we now consider a vector space $\mathcal{V} = \mathbb{C}^d$, with $U(d)$ the corresponding unitary group. 
As discussed above, we have the natural representation $(\tau, \mathcal{V})$ of $U(d)$ on $\mathcal{V}$, where $\tau$ is simply the identity map. 
This representation $\tau$ is in fact irreducible. 
We then take the tensor product of $t$ copies of the natural representation, determining the tensor product representation (or \emph{diagonal representation}) $(\tau^t, \mathcal{V}^{\otimes t})$, given by $\tau^t(U) = U^{\otimes t}$. 
This has a decomposition into irreps for $U(d)$, 
\begin{align}\label{eq:tensor irrep decomp}
    \mathcal{V}^{\otimes t} \cong \bigoplus_i V_{i}^{\oplus n_i}, 
\end{align}
where, as above, each $V_i$ is irreducible and $V_i\cong V_j$ iff $i=j$. 
One can show that the frame potential is precisely equal to the sum of the multiplicities in \eqref{eq:tensor irrep decomp}: $F_{\textnormal{Haar}}^{(t)}=\sum_i n_i$. 
In other words, the frame potential is equal to the total number of $U(d)$-irreps of $\mathcal{V}^{\otimes t}$, counted with multiplicity. 
For $t=1$, $\mathcal{V}$ is already irreducible, as mentioned above, so $F_{\textnormal{Haar}}^{(1)} = \sum_i n_i = 1$. 
For $t=2$, $\mathcal{V}^{\otimes 2}$ decomposes into a direct sum of two inequivalent irreps, namely the symmetric and anti-symmetric tensors, so $F_{\textnormal{Haar}}^{(2)}=\sum_i n_i = 2$. 

Now let $G$ be a (finite or compact) subgroup $G\subseteq U(d)$. 
The representation $\mathcal{V}^{\otimes t}$ restricts to a representation of $G$, which then has its own decomposition into irreps for $G$. 
Generalizing the above, the frame potential $F_{G}^{(t)}$ is equal to the total number of $G$-irreps of $\mathcal{V}^{\otimes t}$, counted with multiplicity. 
(This connection may be proven by expressing the frame potential as an inner product of group characters.) 
Then by Lemma~\ref{lemma:frame potential}, we obtain the following well-known result (see e.g., Theorem 3 of Ref. \cite{gross_evenly_2007}, Proposition 3 of Ref. \cite{zhu_clifford_2016}, or Proposition 34 of Ref. \cite{mele_introduction_2023}): 
\begin{proposition}[\cite{gross_evenly_2007,zhu_clifford_2016,mele_introduction_2023}]
\label{prop:designs-equivalent-defns}
Let \(G \subseteq U(d)\) be a (finite or compact) subgroup of the unitary group. The following are equivalent:
\begin{enumerate}
    \item \(G\) is a unitary \(t\)-design.
    \item The \(t\)-copy diagonal action of \(G\), denoted as \(\tau^{t}(G)\), decomposes into the same number of irreps as the \(t\)-copy diagonal action of \(U(d)\). 
    \item \(\tau^{t}(G)\) has the same commutant as \(\tau^{t}(U(d))\).
    \item  $F^{(t)}_{G} = F^{(t)}_{\textnormal{Haar}}$. 
    \label{cond:frame potential}
\end{enumerate}
\end{proposition}
Therefore, checking if a subgroup of the unitary group forms a \(t\)-design reduces to understanding its representation theory. 
In particular, $G$ is a $1$-design if and only if $\mathcal{V}$ is an irreducible $G$-representation, and $G$ is a $2$-design if and only if $\mathcal{V}\otimes \mathcal{V}$ has exactly $2$ irreducible $G$-subrepresentations 
(corresponding to the symmetric and anti-symmetric subspaces). 

\subsection{Weighted unitary t-designs}
\label{sec:weighted unitary designs}
In Sections~\ref{subsec:unitary-designs-intro} and \ref{sec:state designs} above, we discuss (unweighted) unitary $t$-designs and weighted state $t$-designs respectively. 
Here, we consider the case of \emph{weighted unitary} $t$-designs $\mathcal{E} = \{w_U, U\}_{U\in \mathcal{S}}$ \cite{roy_unitary_2009}, 
especially focusing on the case where $\mathcal{S} = G$ is a group. 
Ref.~\cite{kubischta_quantum_2024} considers the case in which the weights $w_U$ arise from an irrep of $G$, referring to them as "twisted unitary $t$-groups" and using them to construct quantum codes. 
However, the weighted $t$-designs obtained in this way do not give "new" orders of design: Lemma 4 of Ref.~\cite{kubischta_quantum_2024} implies that if $\mathcal{E} = \{w_U, U\}_{U\in G}$ is a twisted unitary $t$-group, then $G$ is an unweighted $t$-design. 
This is \emph{not} related to the assumption that the weights are representation-theoretic: as we will show shortly, a similar result holds in general. 

We now give the definition, a direct generalization of Definition~\ref{def:t-designs-equivalent-def} Part 2. 
(Recalling the twirling notation of \eqref{def:twirl}.)

\begin{definition}\label{def:weighted unitary designs} \cite{roy_unitary_2009}
    Let $\mathcal{E} = \{w_U, U\}_{U\in \mathcal{S}}$ be a weighted ensemble of unitaries in $U(d)$. $\mathcal{E}$ is a \emph{weighted unitary $t$-design} if, for all \(X \in \mathcal{L}(\mathcal{H}^{\otimes t})\), 
    \begin{align}
        \Phi_{\mathcal{E}}^{(t)}(X) = \Phi_{\mathrm{Haar}}^{(t)}(X).
    \end{align}
\end{definition}

For convenience of notation, we modify the twirling notation used above, writing  $\Phi_{\mathcal{S},w}^{(t)}(X) = \sum_{U\in \mathcal{S}} w_U U^{\dagger\otimes t}X U^{\otimes t}$. We omit the weighting from the notation in the unweighted or Haar cases. 

In the unweighted case, $\Phi_G^{(t)}$ is the projection onto $\mathcal{L}_G(\mathcal{H}^{\otimes t})$, the $G$-invariant elements of $\mathcal{L}(\mathcal{H}^{\otimes t})$. 

We note that we may always trivially extend $\mathcal{S}$ to a larger set, taking $w_U = 0$ for the newly added unitaries. 
In particular, let $G$ be the closed group generated by $\mathcal{S}$, which is either a finite group or a continuous compact Lie group. 
(If $\mathcal{S}$ generates a finite group, that is $G$. Otherwise, if $\mathcal{S}$ generates an infinite subgroup of $U(d)$, its closure must be a continuous Lie subgroup $G$.) 
Thus we may always think of weighted $t$-designs as corresponding to some finite or compact group $G$. 

\begin{theorem}\label{thm:no weighted designs}
    Let $\mathcal{E} = \{w_U, U\}_{U\in G}$ be a weighted unitary $t$-design, where $G$ is a closed subgroup of $U(d)$. 
    Then $G$ is an unweighted $t$-design. 
\end{theorem}
\begin{proof}
    We prove the contrapositive: assume that the group $G$ is \emph{not} an unweighted $t$-design. 
    We will prove that no weighting leads to a weighted $t$-design. 
    We note that $\mathcal{L}_{G}(\mathcal{H}^{\otimes t})\supseteq \mathcal{L}_{U(d)}(\mathcal{H}^{\otimes t})$.     
    By the assumption that $G$ is not an unweighted $t$-design, this must be a strict containment (or else the two projections would be equal), 
    so there exists $X\in \mathcal{L}(\mathcal{H}^{\otimes t})$ with $\Phi_G^{(t)}(X) = X$ and $\Phi_{U(d)}^{(t)}(X) \neq X$. 
    Then for any weighted ensemble $\mathcal{E} = \{w_U, U\}_{U\in G}$, we have
    \begin{align}
        \Phi_{G,w}^{(t)}(X) = \sum_{U\in G} w_U (U^{\otimes t})^\dagger A U^{\otimes t} = \sum_{U\in G}w_U A = A \neq \Phi_{U(d)}^{(t)}(X).
    \end{align}
    By Definition~\ref{def:weighted unitary designs}, $\mathcal{E}$ is not a weighted unitary $t$-design. 
\end{proof}

Therefore, for a group $G$ that fails to be a $t$-design, we will not be able to "salvage" it into a weighted $t$-design. 
If our goal is to produce higher orders of unitary $t$-design than possible in the unweighted setting, we must necessarily use a universal generating set: 
\begin{corollary}\label{cor:weighted designs universal}
    Suppose that $d$ and $t$ are such that there exists no non-universal, unweighted, group $t$-design in $U(d)$. 
    If $\mathcal{E} = \{w_U, U\}_{U\in \mathcal{S}}$ is a weighted unitary $t$-design, then the corresponding group $G$ (the closure of the group generated by $\mathcal{S}$) is universal. 
\end{corollary}
\begin{proof}
    By the discussion above, we may trivially extend $\mathcal{E}$ to a weighted unitary $t$-design on $G$. 
    Theorem~\ref{thm:no weighted designs} then implies that $G$ is an unweighted unitary $t$-design. 
    By the assumption, $G$ is universal. 
\end{proof}

Ref.~\cite{Kaposi_designs_2024} introduces the notion of \emph{generalized group $t$-designs}, in which we have a sequence $(G_1, \dots, G_\ell)$ of groups, and we require a certain weighting on the set of products $\{g_1\cdots g_\ell: g_i\in G_i\}$ to give a unitary $t$-design. 
Equivalently, $(G_1, \dots, G_\ell)$ is a generalized group $t$-design if and only if
\begin{align}
    \Phi_{U(d)}^{(t)} = \Phi_{G_1}^{(t)}\cdots \Phi_{G_\ell}^{(t)}.
\end{align}
By Corollary~\ref{cor:weighted designs universal}, if we have a generalized group $t$-design of a "new" order $t$ (higher than obtainable with unweighted group $t$-designs), then the closed group $G$ generated by the $G_i$ must be universal. 
This approach is useful, however, as it allows one to apply representation-theoretic tools in a more general context (extending those of Section~\ref{sec:rep theory}). For example, Theorem 1 of Ref.~\cite{Kaposi_designs_2024} gives a method for constructing generalized group $t$-designs using the representation theory of the groups $G_i$. 
Such designs naturally lead to convenient, finite-length quantum circuits for sampling from $t$-designs: simply sample a random unitary from $G_\ell$, then $G_{\ell-1}$, etc. 
Note this sampling is \emph{unweighted} as long as we sample from each $G_i$ independently; the weighting given in Ref.~\cite{Kaposi_designs_2024} is only necessary if we want to view the $t$-design as sampling from the set of products $\{g_1\cdots g_\ell: g_i\in G_i\}$ (which could have fewer than $|G_1|\cdots |G_\ell|$ elements).

\section{Qudits}\label{sec:qudits}
Qudits are \(d\)-dimensional quantum systems whose states can be manipulated for quantum information processing \cite{wang_qudits_2020}. The Hilbert space is \(\mathcal{H} \cong \mathbb{C}^{d}\) and we denote as \(\mathbb{B} = \{ | j \rangle \}_{j=0}^{d-1}\), a canonical basis for this space. For qubit systems, the Pauli matrices form a unitary and orthonormal basis for the space of operators, \(\mathbb{I}, X, Y, Z\). For the case of qudits, one can define \textit{generalized} Paulis, sometimes called the ``clock and shift'' matrices (as will become clear from their action on the canonical basis). These carry a representation of the Heisenberg-Weyl algebra and in this sense are a discretization of the unitary translations generated from the position and momentum operators. We will interchangeably refer to the qudit Pauli elements as either Paulis or Heisenberg-Weyl group elements, where it is understood that they are the same object.

We will often denote the generalized Paulis with a subscript \(d\) to emphasize that they are a \textit{qudit} Pauli and not the standard qubit Paulis to avoid any confusion. Note that for any \(d>2\), an orthonormal basis of operators for \(\mathcal{L}(\mathcal{H})\) can either be unitary (like the qudit Paulis below) or Hermitian (like the Gell-Mann basis), but not both. With this in mind, the qudit Pauli generators are the unitaries defined by
\begin{align}
Z_{d} | j \rangle &= \omega^{j} | j \rangle, \quad \omega = e^{\frac{2 \pi i}{d}},\\
X_{d} | j \rangle &= | j+1 \rangle,
\end{align}
where the kets are indexed by integers modulo $d$. 
For a single qudit, we let 
\begin{align}
    \mathcal{P}_d = \{X_d^i Z_d^j: i,j\in\mathbb{Z}_d\}
\end{align}
be the group generated by $X_d$ and $Z_d$ modulo phases, which we call the Heisenberg-Weyl or Pauli group. 
We note the $d^2$ elements of $\mathcal{P}_d$ are an orthogonal basis for \(\mathcal{L}(\mathcal{H})\). 
The generators satisfy \(\left( X_{d} \right)^{d} = \mathbb{I} = \left( Z_{d} \right)^{d}\) and $X_d Z_d = \omega^{-1}Z_d X_d$. 
The latter is reminiscent of the commutation relations of the (translation operators associated with) position and momentum operators and 
generalizes to 
\cite{nielsen_universal_2002}
\begin{align}\label{eq:qudit pauli relation}
\left( X^{j}_{d} Z^{k}_{d} \right) \left( X^{s}_{d} Z^{t}_{d} \right) = \omega^{ks-jt} \left( X^{s}_{d} Z^{t}_{d} \right) \left( X^{j}_{d} Z^{k}_{d} \right).
\end{align}

\subsection{The qudit Clifford group}

Given the qudit Pauli group \(\mathcal{P}_{d}\), we can define a qudit Clifford group as its normalizer. Namely,
\begin{align}
Cl_{d} = \{ U\in U(d) ~|~ U P U^{\dagger} \in \mathcal{P}_{d} ~~\forall P \in \mathcal{P}_{d}\}/U(1),
\end{align}
so that conjugation by Cliffords takes Paulis to Paulis (modulo phases). 
Note that we define the Clifford group modulo \(U(1)\) to obtain a finite group, by avoiding global phases. 
One may, of course, extend the Pauli and Clifford groups to multiple qudits as well; we refer to those groups as $\mathcal{P}_d^{(n)}$ and $Cl_d^{(n)}$ respectively. 

The qudit Pauli group forms a unitary \(1\)-design in all dimensions \(d \in \mathbb{N}\), which follows from the fact that they are an orthogonal, unitary basis for \(\mathcal{L}(\mathcal{H})\). 
Recall that the multi-\emph{qubit} Clifford group $Cl_2^{(n)}$ is a $3$-design \cite{webb_clifford_2016}. 
The story for qudit Clifford groups is more complicated: for $d>2$ prime, the Clifford group $Cl_d^{(n)}$ is only a $2$-design. 
If the qudit dimension $d$ is composite, then $Cl_d^{(n)}$ as defined here is only a $1$-design. 
However, for $d=p^k$ a prime power, one may define modified Pauli and Clifford groups, with phase structure coming from the finite fields of order $p^k$, and those are $2$-designs \cite{heinrich_stabiliser_2021}: see Section~\ref{sec:galois}. 

We now recall the symplectic characterization of the qudit Pauli and Clifford groups for general $d$. 
Since the case of composite $d$ is rarely discussed, we go into detail here, but restrict to the case of a single qudit for simplicity. We may identify $\mathcal{P}_d$ with vectors in $\mathbb{Z}_d^2$ by the correspondence $X^i Z^j\leftrightarrow \begin{pmatrix}i & j\end{pmatrix}^T$. 
The exponent in the relation \eqref{eq:qudit pauli relation} above then corresponds to a symplectic form $f$ on $\mathbb{Z}_d^2$:  
\begin{align}\label{eq:symplectic form}
    f\left(\begin{pmatrix}j \\ k\end{pmatrix},\begin{pmatrix}s \\ t\end{pmatrix}\right) := \begin{pmatrix}j & k\end{pmatrix}\begin{pmatrix}0 & -1\\1 & 0\end{pmatrix}\begin{pmatrix}s \\ t\end{pmatrix} = ks-jt.
\end{align}
We now consider how the Clifford action translates to the symplectic setting. By definition, the Clifford group $Cl_d$ permutes the Paulis in $\mathcal{P}_d$ by conjugation: this action is fully (and linearly) determined by the images of $X$ and $Z$, so we may map the Clifford action on $\mathcal{P}_d$ to $2\times 2$ matrices over $\mathbb{Z}_d$. 
The Clifford group may not \emph{arbitrarily} permute $X$ and $Z$, however, as it must preserve the relations \eqref{eq:qudit pauli relation}. In particular, the appropriate matrices $M$ must respect the symplectic structure, meaning we must have $f(Mv, Mw) = f(v,w)$ for all $v,w\in\mathbb{Z}_d^2$. Then $M$ must satisfy
\begin{align}
    M^T \begin{pmatrix}0 & -1\\1 & 0\end{pmatrix} M = \begin{pmatrix}0 & -1\\1 & 0\end{pmatrix},
\end{align}
or equivalently
\begin{align}
    \det(M) = wz-xy = 1\textnormal{ for }M = \begin{pmatrix}w & x \\y&z\end{pmatrix}.
\end{align}
Then we may think of the image of the Clifford group $Cl_d$ as either $Sp_2(\mathbb{Z}_d)$ or, in this single-qudit case, $SL_2(\mathbb{Z}_d)$. 
(Note there are no issues with invertibility, despite $\mathbb{Z}_d$ not being a field in general, since we only consider matrices with $\det(M)=1$.)
In summary, we have
\begin{lemma}
    The conjugation action of $Cl_d$ on $\mathcal{P}_d$ is equivalently described by the multiplication action of $Sp_2(\mathbb{Z}_d)=SL_2(\mathbb{Z}_d)$ on $\mathbb{Z}_d^2$. We have the group isomorphism $Cl_d/\mathcal{P}_d\cong Sp_2(\mathbb{Z}_d)$.
\end{lemma}

\subsection{The qudit Galois-Clifford group forms a unitary 2-design}
\label{sec:galois}
In the present work, we are mostly concerned with the "standard" qudit Clifford group $Cl_d$ defined above, namely which arises as the normalizer of the group generated by the ``clock and shift'' matrices $X_d$ and $Z_d$. However, there also exists a different definition of the Clifford group, sometimes called the \emph{Galois-Clifford group}, see Ref. \cite{heinrich_stabiliser_2021} for an excellent exposition. It is important to note that the Galois-Clifford group forms a unitary \(2\)-design for all prime-power dimensions while the "standard" qudit Clifford group does \textit{not} \cite{graydon_clifford_2021,heinrich_stabiliser_2021}.

For the rest of this discussion, we let $p$ be a prime and take $d=p^k$. We focus on the projective Clifford group (i.e.., modulo phases). There exists a field $\mathbb{F}_d$ of order $d$, called a \emph{Galois field}. As an additive group, it is isomorphic to $\mathbb{F}_p^{\oplus k}$. Thus we will view $\mathbb{F}_d$ as a copy of $\mathbb{F}_p^{\oplus k}$ with a nontrivial multiplication operation. 
We will index the basis for our Hilbert space $\mathcal{H}\cong \mathbb{C}^d$ by elements of $\mathbb{F}_d$ rather than $\mathbb{Z}_d$. 
We let $\chi$ be a nontrivial \emph{additive character}, a group homomorphism from $\mathbb{F}_d$ (viewed as an additive group) to $\mathbb{C}$. 
We then define the ``Galois-Pauli'' operators as follows (the analogue of the Heisenberg-Weyl Paulis): for $u,x,z\in\mathbb{F}_d$, 
\begin{align*}
    X(x)\ket{u} &= \ket{u+x},
    \\ Z(z)\ket{u} &= \chi(zu)\ket{u}.
\end{align*}
This naturally induces a variant of the Pauli and Clifford groups: let $GP_d$ be the group of (projective) unitaries generated by the $X(x)$ and $Z(z)$ modulo phases, and let $\widetilde{GC}_d$ be the group of (projective) unitaries normalizing $GP_d$. (Alternatively, one can work non-projectively by choosing the phases very carefully, as in Ref. \cite{heinrich_stabiliser_2021}.) 
 
The Galois-Clifford group $GC_d$ is a subgroup of $\widetilde{GC}_d$, chosen so that $GC_d/GP_d\cong Sp_2(\mathbb{F}_d)$, through a homomorphism similar to the one discussed for $Cl_d$ above. 
This is the essential distinction: on the relevant Pauli groups, $GC_d$ acts by $Sp_2(\mathbb{F}_d)$, whereas $Cl_d$ acts by $Sp_2(\mathbb{Z}_d)$. 
In particular, when $d=p$ is prime, these constructions are all equivalent and we have $GC_p = \widetilde{GC}_p = Cl_p$. 

One may carry out similar constructions for multi-qupit systems. Letting $GC_{d,m}$ be the Clifford group for $m$ $d$-dimensional qudits, there is a natural embedding 
\begin{align*}
    GC_{p^k,m} \hookrightarrow GC_{p,km},
\end{align*}
and in fact $\widetilde{GC}_{p^k,m}\cong GC_{p,km} = \widetilde{GC}_{p,km}$ \cite{heinrich_stabiliser_2021}. The distinction between "which Clifford group" to choose in the qudit case has been the source of confusion about the (non)existence of unitary \(2\)-designs, e.g., in Refs. \cite{Jafarzadeh2020,Goswami2021}.

\subsection{Exact qudit designs: the Wootters and Fields construction}
While (unweighted) unitary designs exist for all \(t,d \in \mathbb{N}\) \cite{roy_unitary_2009}, explicit examples of these are only known in special cases. It is, a priori, entirely unclear how one would go about searching for them. This is where the connection to group-theory has been quite successful. By utilizing the representation theory of certain groups and leveraging character tables, a systematic classification of unitary designs in various dimensions was obtained in \cite{bannai_unitary_2018}. For quantum information, the main takeaway from this classification is that for all prime-power dimensions, the Clifford group forms a unitary \(2\)-design. In fact, the only known family (as a function of the dimension) of unitary designs for \(t>1\) is the Clifford group.

The respite offered by group-theoretic methods is both a boon and a bane. For e.g., the results of \cite{bannai_unitary_2018} also tells us that so-called "unitary \(t\)-groups" do not exist for non-prime-power dimensions such as \(d=6,10\). Of course one can always find \textit{sporadic} groups for these dimensions, but we run out of them as well. As a result, it is not obvious how to perform quantum information primitives that rely on the unitary \(2\)- or \(3\)-design property such as RB and classical shadow tomography. Our weighted unitary \(2\)- and \(3\)-designs allow one to circumvent these issues. We summarize in \cref{table:examples-dimensions}, known families of qudit unitary and state designs.

In Ref. \cite{wootters_optimal_1989}, Wootters and Fields introduced the following set of states (following the notation in Ref. \cite{iosue_continuous-variable_2024}),
\begin{align}
| q_{1}, q_{2} \rangle \equiv \frac{1}{\sqrt{d}} \sum_{n=0}^{d-1} \exp \left[\frac{2 \pi \mathrm{i}}{d}\left(q_{1} n+q_{2} n^{2}\right)\right]|n\rangle
\end{align}
For each \(q_{2} \in\{0, \ldots, d-1\}\) it is easy to show that \(\left\{\left|q_{1},q_{2}\right\rangle \mid q_{1} \in\{0, \ldots, d-1\}\right\}\) forms an orthonormal basis. Moreover, if \(d\) is prime, then the \(d^2\) states above, along with the computational basis form an exact \(2\)-design; and also a set of \(d+1\) MUBs.

 In Ref. \cite{iosue_continuous-variable_2024}, the authors generalized the discrete set of states here to a \textit{continuous} design for an infinite-dimensional (separable) Hilbert space. They introduce the \textit{non-normalizable} states,
\begin{align}
|\theta,\phi \rangle := \frac{1}{\sqrt{2 \pi}} \sum_{n \in \mathbb{N}_{0}} \exp \left[\mathrm{i}\left(\theta n+\varphi n^{2}\right)\right]|n\rangle.
\end{align}
They go on to show that the Fock basis for \(\mathcal{H}\) along with the states \(\{ | \theta,\phi \rangle \}\) which are now continuously parametrized, when integrated over \(\theta,\phi\) form a ``rigged \(2\)-design''. More formally, let \(\Pi_{2} := \frac{1}{2} \left( \mathbb{I} + \mathbb{S} \right)\) be the symmetric projector on \(\mathcal{H}^{\otimes 2}\), then,
\begin{align}
\label{eq:infinite-dimensional-2-design}
\Pi_{2} = \frac{1}{2} \left( \sum\limits_{n \in \mathbb{N}_{0}}^{} | n \rangle \langle  n | ^{\otimes 2} \right) + \frac{1}{2} \int\limits_{-\pi}^{\pi} d\phi \int\limits_{-\pi}^{\pi} d \theta \left( | \theta,\phi \rangle \langle  \theta,\phi |\right) ^{\otimes 2} . 
\end{align}

One of the key results of our work stems from the observation that, the (continuous set of) states above, when appropriately redefined for a \textit{finite} \(d\)-dimensional Hilbert space, form a weighted \(2\)-design for \textit{any} \(d \in \mathbb{N}\). Perhaps even more importantly, this finite-dimensional construction does not suffer from any of the issues that arise in the infinite-dimensional case, which is what led the authors of Ref. \cite{iosue_continuous-variable_2024} to introduce rigged \(t\)-designs and so on.

 More formally, let \(\mathcal{H} \cong \mathbb{C}^{d}\) be an arbitrary finite-dimensional Hilbert space with \(\mathbb{B} = \{ | n \rangle \}_{n=0}^{d-1}\), an orthonormal basis (which we refer to as the `Fock' basis). Define the `phase' states (w.r.t. the Fock basis) as,
\begin{align}
\label{eq:phase-states-qudits}
|\theta, \phi\rangle:=\frac{1}{\sqrt{d}} \sum_{n=0}^{d-1} \exp \left[i\left(\theta n+\phi n^{2}\right)\right]|n\rangle,
\end{align}
where \(\theta, \phi \in [-\pi,\pi)\). Then the Fock states along with the phase states form a \textit{weighted} \(2\)-design for \(\mathcal{H}\). We summarize this in the following observation.
\begin{proposition}
The following ensemble of states forms a continuous weighted state 2-design for any finite-dimensional qudit.
\begin{align}
\Pi_{2} = \frac{1}{2} \left( \sum_{n=0}^{d-1} | n \rangle \langle  n | ^{\otimes 2} \right) + \frac{d^2}{2} \int\limits_{-\pi}^{\pi} d\phi \int\limits_{-\pi}^{\pi} d \theta \left( | \theta,\phi \rangle \langle  \theta,\phi |\right) ^{\otimes 2}. 
\end{align}
\end{proposition}
Notice that, unlike the projector in the infinite-dimensional case, \cref{eq:infinite-dimensional-2-design}, we have a prefactor of \(d^2/2\) in front of the phase states. This is the reason why the states above form a \textit{weighted} \(2\)-design and not a \textit{uniform} \(2\)-design.

\subsection{Generating weighted state designs from uniform designs}
In Ref. \cite{zhu_clifford_2016}, the following informal theorem was implied: Given a (uniform) \(t\)-design, \(\mathcal{E} = \{ | \psi_{j} \rangle \}\) in \(d\)-dimensions, it induces a \textit{weighted} \(t\)-design for all \(\tilde{d} \leq d\). Define the ensemble \(\tilde{\mathcal{E}}\) as,
\begin{align}
\tilde{\mathcal{E}} = \{ \tilde{w_{j}}, | \tilde{\psi_{j}} \rangle \}, \quad | \tilde{\psi_{j}} \rangle := \frac{1}{\left\Vert P | \psi_{j} \rangle \right\Vert_{}^{}} | \psi_{j} \rangle \quad \text{and} \quad  \tilde{w_{j}} = \left\Vert P | \psi_{j} \rangle \right\Vert_{}^{}.
\end{align}
Here, \(P\) is a projector onto \(\mathbb{V} \subset \mathcal{H}\), a \(\tilde{d}\)-dimensional subspace of \(\mathcal{H}\). Then, \(\tilde{\mathcal{E}}\) is a weighted \(t\)-design. Unfortunately, the above claim is false. We provide a simple modification to fix this issue\footnote{To the best of our knowledge, the first such construction was observed in Ref. \cite{zhu_clifford_2016}. However, the result as stated there has typos/mistakes. The authors independently realized how to (simply) fix this but would like to point to the original (yet unpublished) paper. In particular, the authors would like to thank Huangjun Zhu for private communication on this topic.}. \\

\begin{theorem}[Weighted state designs via projections \cite{zhu_clifford_2016}]
\label{thm:weighted-design-from-uniform}
Given a uniform state \(t\)-design, \(\mathcal{E} = \{ | \psi_{j} \rangle \}\) in a \(q\)-dimensional space, it induces a \textit{weighted} \(t\)-design for all \(d \leq q\). The weighted state \(t\)-design can be obtained as \(\tilde{\mathcal{E}} = \{ \tilde{w_{j}}, | \tilde{\psi_{j}} \rangle \}\), where,
\begin{align}
| \tilde{\psi_{j}} \rangle := \frac{1}{\left\Vert P | \psi_{j} \rangle \right\Vert_{}^{}} P | \psi_{j} \rangle, \tilde{w_{j}} = \nu \left\Vert P | \psi_{j} \rangle \right\Vert_{}^{2t}, \text{ and } \nu = \frac{1}{\sum\limits_{j=1}^{\left| \mathcal{E} \right|} \left\Vert P | \psi_{j} \rangle \right\Vert_{}^{2t}}.
\end{align}
\end{theorem}

A simple calculation reveals that the chosen weights cancel out the normalization of the new states \(| \tilde{\psi_{j}} \rangle \in \mathbb V\), i.e.,
\begin{align}
\label{eq:simplify-weighted-design-to-uniform-design}
\sum\limits_{j=1}^{\left| \mathcal{E} \right|} \tilde{w_{j}} | \tilde{\psi_{j}} \rangle \langle  \tilde{\psi_{j}} |^{\otimes t} =  \frac{1}{\sum\limits_{j=1}^{\left| \mathcal{E} \right|} \left\Vert P | \psi_{j} \rangle \right\Vert_{}^{2t}} \sum\limits_{j=1}^{\left| \mathcal{E} \right|} \left( P | \psi_{j} \rangle \langle  \psi_{j} | P \right)^{\otimes t}. 
\end{align}
Now it remains to be shown that this is equal to the symmetric projector on \(\mathbb{V}^{\otimes t}\). Let us see how this works for the case of \(t=1\), the generalization to larger \(t\) is straightforward due to multilinearity. For \(t=1\), let us consider a uniform \(1\)-design for \(\mathcal{H}\). Of course we can choose a \(t>1\) design but this is sufficient for illustrative purposes. Since any orthonormal basis of \(\mathcal{H}\) suffices to form a \(1\)-design, we have \(\left| \mathcal{E} \right| = d\) and  \(\sum_{j=1}^{\left| \mathcal{E} \right|}  | \psi_{j} \rangle \langle  \psi_{j} | = \mathbb{I} \). As a result, \cref{eq:simplify-weighted-design-to-uniform-design} simplifies to, \(\sum_{j=1}^{\left| \mathcal{E} \right|} P | \psi_{j} \rangle \langle  \psi_{j} | P = P \mathbb{I} P = P\). But notice that \(P |_{\mathbb{V}} = \mathbb{I}_{\mathbb{V}}\) that is, the projector restricted to the subspace \(\mathbb{V}\) is simply the identity operator on the subspace. The normalization can be simplified as
\begin{align}
\sum_{j=1}^{\left| \mathcal{E} \right|} \left\Vert P | \psi_{j} \rangle \right\Vert_{}^{2} = \sum_{j=1}^{\left| \mathcal{E} \right|} \operatorname{Tr}\left[ P | \psi_{j} \rangle \langle  \psi_{j} | P \right] = \operatorname{Tr}\left[ P \right] = \tilde{d}.
\end{align}
Therefore, the RHS of \cref{eq:simplify-weighted-design-to-uniform-design} is equal to \(P/\tilde{d}\) as expected.

An immediate consequence of this projection formula is that it immediately allows for qudit \(3\)-designs in arbitrary dimensions. This follows from the fact that multiqubit Clifford group forms a \(3\)-design and choosing a projector onto the \(d < 2^n\)-dimensional subspace.

\subsection{Exact weighted \texorpdfstring{\(2\)}{2}-design in arbitrary dimensions}
Notice that this is one of the few cases where a \(2\)-design is generated by states that do \textit{not} come from the (orbit of the) Clifford group. Let us try to construct an example. Let \(p > 2\) be an odd prime. Then, the Wootters-Fields construction gives us a \(2\)-design in every \(p\)-dimension. The elements of this ensemble are,
\begin{align}
\mathcal{E} = \{ | 0 \rangle, | 1 \rangle, \cdots, | p-1 \rangle \} \cup \{ | q_{1}, q_{2} \rangle \}_{q_{1}, q_{2} = 0}^{p-1},
\end{align}
with \(\left| \mathcal{E} \right| = p(p+1)\). Since the projector in the \cref{thm:weighted-design-from-uniform} above is completely general, we can simply choose it to be of the form, \(P = \sum\limits_{n=0}^{d-1} | n \rangle \langle n |\), where \(d = \mathrm{dim}(\mathbb{V}), d \leq p\) is \textit{not} necessarily prime. This allows us to formulate the following result.\\

\begin{theorem}[Exact weighted state \(2\)-design in any qudit dimension]
Let \(p, d \in \mathbb{N}\) and \(p \geq d\) be an odd prime. Then the following ensemble forms a weighted \(2\)-design for every qudit dimension \(d\). 
\begin{align}
&\mathcal{E} = \{ | 0 \rangle, | 1 \rangle, \cdots, | d-1 \rangle \} \cup \{| \widetilde{\theta,\phi} \rangle\}_{\theta,\phi}, \text{ where } | \widetilde{\theta,\phi} \rangle := \frac{1}{\sqrt{d}} \sum_{n=0}^{d-1} \exp \left[ \frac{2 \pi i}{p} \left( \theta n + \phi n^{2} \right) \right] | n \rangle,\\ \nonumber
&\text{and } \theta,\phi \in \{ 0,1, \cdots, p-1 \}.
\end{align}
\end{theorem}
Note here that \(\theta,\phi \in \{ 0,1, \cdots, p-1 \}\) so we still have \(p^{2}\) such elements. That is, the size of this ensemble is \(p^{2} + d \geq d (d+1)\) and in general is \textit{not} the optimal ensemble. The number of extra elements (on top of the optimal bound) is \(p^{2} - d^{2}\).

The weights associated with this ensemble are,
\begin{align}
w_{\mathrm{Fock}} = \frac{1}{d (d+1)} \quad \text{and} \quad w_{\mathrm{Phase}} =  \frac{d}{p^{2} (d+1)}.
\end{align}

As an aside note that this means that in \textit{any} truncated Fock state simulation, we always have access to a \(t=3\) design and do not need to use the rigged design construction of \cite{iosue_continuous-variable_2024}.

\subsubsection{Example of weighted state \(2\)- and \(3\)-designs for \(d=6\).}
The smallest dimension where the unitary \(t\)-group structure restricts the existence of a design is \(d=6\), although, this is one of those cases where there exist "sporadic groups," and therefore \(d=6\) does have a unitary group \(2\)-design but not a unitary group \(3\)-design \cite{bannai_unitary_2018}. Nonetheless, we can construct explicit examples of weighted \(2\)- and \(3\)-designs. For the \(2\)-design case we simply list the states from the rigged design scenarios. For the \(3\)-design case, we construct them from the Clifford group for \(d=2^3\). Notice that the Clifford group (and hence stabilizer states) form an exact multiqubit \(3\)-design in dimensions \(2^n\) \cite{zhu_clifford_2016,webb_clifford_2016,kueng_qubit_2015}. Constructing weighted \(3\)-designs can become nontrivial since the number of stabilizer states for \(n\) qubits grows exponentially as \(2^{(n+1)(n+2)/2}\) \cite{gross_hudsons_2006} and one needs to compute inner products between these states and the projector to use the projection lemma. However, since one can efficiently estimate the inner products between stabilizer states \cite{https://doi.org/10.48550/arxiv.1210.6646}, we can easily construct weighted state \(3\)-designs in any dimension \(d\), albeit not in an optimal way.

\section{Practical applications}
\label{sec:Application to spin and cavity-QED qudits}

\subsection{Spin-coherent states and the manipulation of high-spin nuclei}
\label{subsec:spin-coherent}
Here we briefly introduce the framework of spin-coherent states (SCSs) following the excellent reference (modulo some typos) \cite{goldberg_extremal_2020}. A variant of the discussion here can also be found in the supplemental material of Ref. \cite{yu2024schrodingercat} by a subset of the authors. We will focus on a  \((2S+1)\)-dimensional Hilbert space, \(\mathcal{H}_{S}\) which can host a irreducible representation of \(SU(2)\)\footnote{In \cite{mourik_exploring_2018} they use \(I\) for the spin quantum number that we refer to as \(S\) instead}. That is, we need a irreducible representation of the Lie algebra \(su(2)\) defined by the generators, \(\left\{\hat{S}_{x}, \hat{S}_{y}, \hat{S}_{z}\right\}\) and the canonical commutation relation,
\begin{align}
\left[\hat{S}_{x}, \hat{S}_{y}\right]=i \hat{S}_{z},
\end{align}
and cyclic permutations thereof. The Casimir operator \(S^{2} = S^{2}_{x} + S^{2}_{y} + S^{2}_{z} = S(S+1) \mathbb{I}\) commutes with all elements of the Lie algebra. Recall that the Hilbert space is spanned by the (simultaneous) eigenbasis of \(S^{2}, S_{z}\), i.e.,
\begin{align}
&\mathcal{H}_{S} = \mathrm{span} \{ | S,m \rangle ~|~ m = -S, \cdots, S \},\\
& S^{2} | S,m \rangle = S(S+1) | S,m \rangle \\
& S_{z} | S,m \rangle = m | S,m \rangle.
\end{align}
The raising and lowering operators act as,
\begin{align}
\label{eq:raising-lowering-defn}
\hat{S}_{+}|S, m\rangle & =\sqrt{(S-m)(S+m+1)}|S, m+1\rangle \\ \hat{S}_{-}|S, m\rangle & =\sqrt{(S+m)(S-m+1)}\rangle|S, m-1\rangle
\end{align}

The phase space is \(SU(2)/U(1) \cong S^2\), i.e., the Bloch sphere. Let us define the displacement operator as
\begin{align}
\hat{D}(\theta, \phi)= \exp \left[ \frac{\theta}{2} \left( e^{i \phi} S_{-} - e^{-i \phi} S_{+} \right) \right] .
\end{align}
If we define \(\mathbf{n}=(\sin \theta \cos \phi, \sin \theta \sin \phi, \cos \theta)\) then, \(SU(2)\) coherent states can be equivalently defined as
\begin{align}
|\mathbf{n}\rangle=\hat{D}(\mathbf{n})|S, S\rangle,
\end{align}
where \(\theta,\phi\) represent the spherical coordinates of the unit vector \(\mathbf{n}\) on the Bloch sphere. Using a stereographic projection, we can define the coordinates
\begin{align}
\zeta=\tan (\theta / 2) e^{i \phi}, \quad \eta=\ln \left(1+|\zeta|^{2}\right).
\end{align}
This lets us define \(SU(2)\) coherent states in an analogous form
\begin{align}
|\mathbf{n}\rangle=\frac{1}{\left(1+|\zeta|^{2}\right)^{S}} \exp \left(\zeta \hat{S}_{-}\right)|S, S\rangle.
\end{align}

Some important properties of SCS include being eigenstate of the spin operator pointing in the direction \(\mathbf{n}\) on the Bloch sphere: \(\hat{S}_{\mathbf{n}}=\hat{\mathbf{S}} \cdot \mathbf{n}\) then \(\hat{S}_{\mathbf{n}}|\mathbf{n}\rangle=-S|\mathbf{n}\rangle\). Non-orthogonality between different SCS: \(\left|\left\langle\mathbf{n}_{1} \mid \mathbf{n}_{2}\right\rangle\right|^{2}=\left[\frac{1}{2}\left(1+\mathbf{n}_{1} \cdot \mathbf{n}_{2}\right)\right]^{2 S}\). And a continuous resolution of identity: 
\begin{align}
\frac{1}{4 \pi} \int_{S^{2}} d \mathbf{n}|\mathbf{n}\rangle\langle\mathbf{n}|=\frac{\mathbb{I}}{d},
\end{align}
where \(d = 2S+1\). In particular, the latter shows that SCS form a \textit{continuous} \(1\)-design, where $d\mathbf{n} = \sin(\theta)d\theta d\phi$.

\subsection{Experimental platforms for qudits}
Qudits may seem relatively rare in practical quantum information processing, but they are ubiquitous in the physical world \cite{wang_qudits_2020}. Many qubit platforms are, in fact, the truncation to the lowest energy doublet of more complex, high-dimensional systems. The steady progress in understanding the computational value of high-dimensional quantum information \cite{kiktenko2025colloquium} has motivated revisiting certain platforms to embrace, rather than avoid, the existence of more than two quantum levels. Some of the earliest qudit demonstrations originate from atomic and optical physics. Sophisticated qudit control has been demonstrated with cesium \cite{PhysRevX.11.021010} and dysprosium atoms \cite{Chalopin2018}, and more recently in trapped ions \cite{Hrmo2023}. Photonics offers a wide range of high-dimensional quantum states, such as high orbital angular momentum states \cite{malik2016multi}, or spatial spatial \cite{wang2018multidimensional}, temporal \cite{martin2017quantifying} and frequency \cite{kues2017chip} modes. Natural examples of physical qudits are spins with quantum number $S>1/2$, resulting in Hilbert space dimension $d=2S+1$. Such high-spin systems naturally occur in nuclear or electronic spins in the solid state \cite{asaad2020coherent,yu2024schrodingercat,fu2022experimental,soltamov2019excitation}, and magnetic molecules \cite{thiele2014electrically}. 

Qudit can also be engineered in artificial quantum systems. For instance, the transmon is the workhorse of superconducting quantum processors. In its common use as a qubit, a transmon is the truncation to the lowest two energy states of an anharmonic oscillator. However, the parameters of the transmon can be optimized to facilitate the addressing of more than just two states \cite{bianchetti2010control,blok2021quantum}, all the way to the simulation of an 8-dimensional spin \cite{champion2024multi}. A harmonic oscillator can also be viewed as a qudit with $d\rightarrow \infty$. This way of looking at a harmonic oscillator gives rise to what is known as continuous variables quantum computing \cite{braunstein_quantum_2005}. It is also possible to remain within the realm of discrete variables quantum information, by coupling a few of the Fock states to an ancilla transmon to control the transition matrix elements between the Fock states \cite{roy2024synthetic}. An interesting new direction is the creation of logical finite-dimensional qudits, on top of states defined by continuous variables such as the Gotteman-Kitaev-Preskill (GKP) ones \cite{brock_quantum_2025}. Curiously, the seminal GKP proposal itself \cite{gottesman_encoding_2001} began by considering the encoding of a qudit, with the qubit being a special case.    

A minimal Hamiltonian that allows universal qudit control may take the form:
\begin{align}
    \mathcal{H}_0 = \alpha S_z + \beta S_z^2 + \mathcal{H}_{\rm control}(t).
\end{align}
This form results in a natural basis $\mathbb{B}=\{\ket{m}\}_{m=-S}^S$, where $\ket{m}$ are eigenvectors of $S_z$ and, in the case of an actual spin system, represent the projections along the quantization axis. If $|\alpha| \gg |\beta|$, the energy spacing between each pair of adjacent eigenstates $\ket{m},\ket{m-1}$ is $\approx \alpha + 2\beta(m-1/2)$. For simplicity we assume that the operator norm of \(H_{\rm control}\) is small compared to the original Hamiltonian \(H_0\). For spin systems, the linear term $\alpha S_z$ is usually a Zeeman energy, i.e. the interaction of the spin with a static magnetic field $B_0$. The quadratic term $\beta S_z^2$ depends on the details of the physical system. For nuclear spins in the solid state, it may originate from the electric quadrupole interaction, that couples the nuclear electric quadrupole moment to an electric field gradient \cite{asaad2020coherent}; for electronic spins, it may arise from the combination of crystal fields and spin-orbit interaction \cite{thiele2014electrically}; for cold atoms, it may be the a.c. Stark shift from a detuned laser field \cite{PhysRevLett.114.240401}; for transmons, it is a function of the anharmonicity of the confining potential \cite{champion2024multi}, etc.  

The control Hamiltonian $\mathcal{H}_{\rm control}(t)$ is often the one that takes the most diverse forms. Taking the example of a spin $S>1/2$, $H_{\rm control}$ will usually be an oscillating magnetic field $B_1$ orthogonal to the quantization axis, so that:
\begin{align}
      \mathcal{H}_\mathrm{control}(t) = \gamma S_x \sum_{k=1}^{2S}\cos(2\pi f_k t + \phi_k)B_{1,k}(t),
  \label{eq:H_1}
\end{align}
where $\gamma$ is a gyromagnetic ratio, and we have explicitly accounted for the possibility that the control Hamiltonian may consist of multiple (possibly simultaneous) tones at the frequencies $f_k$ corresponding to transitions between eigenstates of $\mathcal{H}_0$. The existence of $2S=d-1$ spectrally resolved transition frequencies is due to the nonlinear term $\beta S_z^2$ in $\mathcal{H}_0$, and is the fundamental property that allows universal control and the creation of nonclassical states in qudit systems \cite{yu_schrodinger_2025,vaartjes_certifying_2025}. The remarkable progress in classical control electronics in the last decade has made it possible to efficiently produce Hamiltonians like Eq.~\ref{eq:H_1}, and has provided further impetus to the experimental development of qudit platforms.

\subsection{SNAP gates and qudit universality}
\label{sec:snap-universality}
In the following, we will refer to the Hilbert space dimension as \(d = 2S+1\). And since we work with a fixed total spin \(S\), we will refer to the \(| S,m \rangle\) states as \(| m \rangle\) instead. With this notation, the SNAP gate is simply a diagonal unitary with arbitrary phases on each diagonal (inspired from the analogous quantity in cavity-QED systems \cite{krastanov_universal_2015}),
\begin{align}
S(\vec{\theta})=\sum_{m=1}^{d} e^{i \theta_{m}}|m\rangle\langle m|,
\end{align}
where \(\vec{\theta} = \left( \theta_{1}, \theta_{2}, \cdots, \theta_{d} \right)\) is the vector of phases. Occasionally we will need to use the ``linear SNAP,'' \(S(\phi) = \sum\limits_{m=1}^{d} e^{i m \phi} | m \rangle \langle  m |\). We call this a linear SNAP because of the linear relationship between the phases on the diagonal, \(\phi_{m} = m \phi ~~\forall m \in \{ 1,2, \cdots, d \}\). We first note the following simple identity,
\begin{align}
D(\theta,\phi) = S(\phi) D(\theta,0) S^{\dagger}(\phi),
\end{align}
where \(S(\phi)\) is the linear SNAP gate defined above. Moreover, notice that since, \(D(\theta,0) = \exp \left[ \frac{\theta}{2} (S_{-} - S_{+})  \right]\) and \(S_{-} - S_{+} = -2i S_{y}\), we have, \(D(\theta,0) = e^{-i \theta S_{y}}\). That is, displacements by a \(\theta\) direction only correspond to a \(S_{y}\) rotation with the same angle (up to a minus sign). Moreover, we have,
\begin{align}
D(\theta,\phi) = S(\phi) \exp \left[ -i \theta S_{y} \right]  S^{\dagger}(\phi).
\end{align}
That is a displacement of \((\theta,\phi)\) can be generated from linear SNAP gates of amount \(\phi\) and a \(S_{y}\) rotation. 

Our proof strategy follows closely that of Ref. \cite{krastanov_universal_2015}. Using the fact that, \(D(\theta, \phi)=S(\phi) D(\theta, 0) S^{\dagger}(\phi)\) we will only focus on rotations generated by \(D(\theta, 0)=e^{-i \theta S_{y}}\) with \(\theta \in [0,\pi)\). Let \(Q_{n}=\sum_{m=0}^{n}|m\rangle\langle m|\) be the projector onto the first \(n\) levels. Then, it is easy to show that,
\begin{align}
J_{n}:=-i \left[S_{y}, Q_{n}\right] \propto(|n\rangle\langle n+1|+| n+1\rangle\langle n|).
\end{align}
That is, \(J_{n}\) is the generator of \(SO(2)\) rotations between levels \(|n\rangle\) and \(|n+1\rangle\) Therefore, using \(J_{n}\) and \(Q_{n}\) we can generate the Lie algebra \(\mathfrak{u}(d)\) and hence generate any Hamiltonian evolution for the qudit, proving the universality of SNAP and displacement operations.

Moreover, similar to \cite{krastanov_universal_2015}, let $R_n(\epsilon) = e^{i\epsilon}Q_n + (I - Q_n)$? That is, it acts as a `1' on states $|m>n\rangle$, but otherwise with the phase. Then this is a snap gate as required, of the form $S(\epsilon, \dots, \epsilon, 0, \dots, 0)$. Then the `group commutator,'
\begin{align}
\hat{D}(\epsilon) \hat{R}_{n}(\epsilon) \hat{D}(-\epsilon) \hat{R}_{n}(-\epsilon)=\exp \left(i J_{n} \epsilon^{2}+\mathcal{O}\left(\epsilon^{3}\right)\right).
\end{align}
That is, switching on and off the (infinitesimal) SNAP and displacement gates above generates (infinitesimal) rotations between levels \(| n \rangle\) and \(| n+1 \rangle\). Although not optimal, this allows for a protocol to implement nontrivial \(\theta \gg 0\) rotations by breaking them into infinitesimal rotations and then using the alternating sequence above.

This combined with the fact that we can compile any \(d \times d\) unitary into \(O(d^{2})\) SNAPs and \(SO(2)\) rotations (also known as Givens rotations) \cite{krastanov_universal_2015} allows for a protocol to implement any single qudit unitary gate.

\begin{figure}
    \centering
    \includegraphics[scale=0.25]{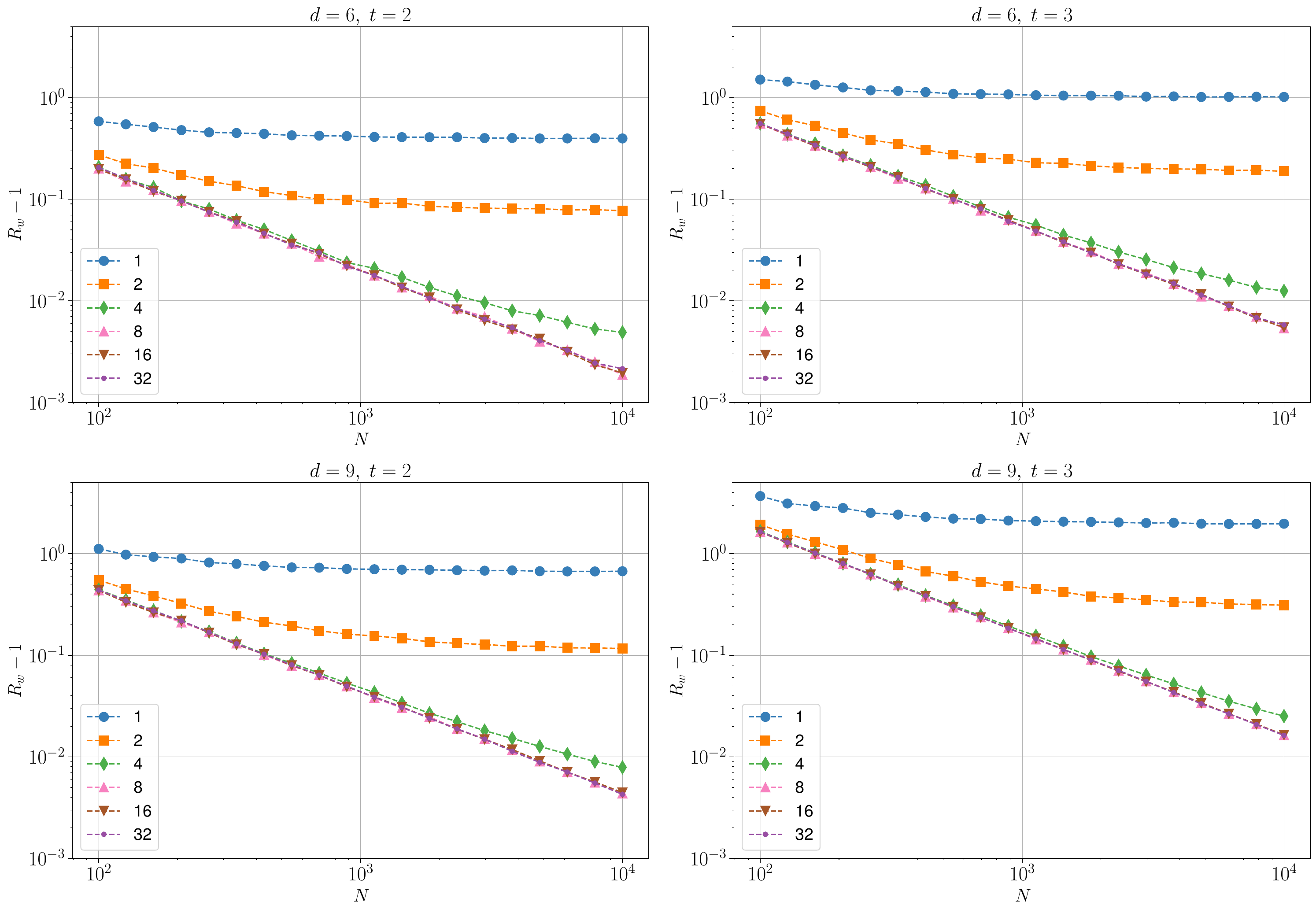}
    \caption{State Welch test ratio $R_w$ (ratio of left to right hand side in Eq.~\eqref{eq:welch-test-weighted-design}) minus one, where $R_w-1=0$ for a state $t$-design. The four different panels show different qudit dimension $d=6,9$, and the approach to a $t=2,3$ design, with the number of states $N$ used. The legend gives the circuit depth, where each layer is a random real displacement $D(\theta, \phi=0)$ ($\theta\in [0,\pi]$ selected according to density $\frac{1}{2}\sin \theta$), followed by a random SNAP gate (angles selected uniformly in $[0, 2\pi)$). The initial state is $|S,S\rangle$. For deep enough circuits, the curves fall like $R_w-1\sim 1/N$.}
    \label{fig:disp-snap-t}
\end{figure}

\subsection{\(SU(2)\) covariant rotations are not a unitary 2-design}
\label{sec:scs not design}
Recall the notation from~\cref{sec:Application to spin and cavity-QED qudits}. Let $S\geq 1$ be a half-integer and consider the representation of $SU(2)$ on the $d$-dimensional Hilbert space $\mathcal{H}_S$. Let $G\subseteq U(\mathcal{H}_S)\cong U(2S+1)$ be the image of $SU(2)$. We will show $G$ is not a $2$-design in $U(\mathcal{H}_S)$. We give 3 related proofs of this fact, each more explicit than the last. For the first, we recall the following theorem. 

\begin{theorem}[\cite{Saied2024,katz2004larsen}]
\label{thm:no-cts-2-design}
Let $\mathcal{V}\cong\mathbb{C}^d$, with $d\geq 2$. Let $G$ be an infinite closed subgroup of $U(d)=U(\mathcal{V})$. If $G$ is a unitary $2$-design, then $SU(d)\subseteq G$. In particular, any set of unitaries generating $G$ is a universal set.
\end{theorem}

The \textit{colloquial} way to interpret this theorem is as follows: continuous subgroups of the full unitary group only come in two flavors: (i) either they form a unitary \(1\)-design or (ii) they are universal and hence form a unitary \(t\)-design for any \(t \in \mathbb{N}\). The latter is not particularly useful since the primary goal of unitary designs is to replace integration over the full unitary group with a sum over a \textit{discrete} set. The former is also limited in its utility since a unitary \(1\)-design can simply be obtained by an orthonormal unitary basis of the operator space. In practice, the \(t\geq 2, t\in \mathbb{N}\) designs are most useful, e.g., for randomized benchmarking and shadow tomography. This limits the applicability of \textit{continuous} groups as unitary designs.

We now return to our consideration of the group $G$ given above, the image of $SU(2)$ in $U(2S+1)$. Since $G$ is not universal in $U(2S+1)$ for $S\geq 1$, the theorem implies it cannot be a $2$-design. 

For the second proof that $G$ is not a $2$-design, we apply Proposition~\ref{prop:designs-equivalent-defns}, showing that $\mathcal{H}_S\otimes \mathcal{H}_S$ has at least 3 $G$-irreps for $S\geq 1$. 
In Appendix~\ref{sec:character appendix}, we use character theory to show that $\mathcal{H}_S\otimes \mathcal{H}_S^*$ has $2S+1\geq 3$ distinct irreps, each of multiplicity $1$. As in \cite{Saied2024}, one can show that $\mathcal{H}_S\otimes \mathcal{H}_S^*$ and $\mathcal{H}_S\otimes \mathcal{H}_S$ have the same number of irreps (or apply the same character theory used in Appendix~\ref{sec:character appendix} to prove a similar decomposition for the latter), leading to the conclusion that $G$ is not a $2$-design. 

Finally, we give a more direct proof, still applying Proposition~\ref{prop:designs-equivalent-defns}. 
Consider the states
\begin{align}
    \ket{S,S-2}\otimes \ket{S, S},\,\, \ket{S, S-1}\otimes\ket{S, S-1}, \,\,\ket{S, S}\otimes\ket{S, S-2}\in \mathcal{H}_S\otimes \mathcal{H}_S.
\end{align}
Under the action of $S_z$ on $\mathcal{H}_S\otimes \mathcal{H}_S$, namely $S_z \mapsto S_z\otimes I + I\otimes S_z$, these states are linearly independent and all have eigenvalue $2S-2$. However, in any irrep of $SU(2)$, eigenspaces of $S_z$ are all $1$-dimensional. 
Therefore, $\mathcal{H}_S\otimes \mathcal{H}_S$ must decompose into a sum of at least $3$ irreps. 
We conclude that $G$ is not a $2$-design. 

It is worth emphasizing that such covariant \(SU(2)\) rotations are the native gates used to perform quantum computing via high-spin qudits \cite{yu2024schrodingercat}.

\subsection{Spin coherent states do not form a state 2-design}
\begin{lemma}[\cite{zhu_clifford_2016}]
If a group \(G \leq U(d)\) forms a unitary \(t\)-design then any orbit of a normalized vector under the action of \(G\) forms a complex projective \(t\)-design (aka a state \(t\)-design).
\end{lemma}

Notice that this result, combined with the proof above that \(SU(2)\) covariant rotations do not form a group 2-design immediately tells us that averaging over SCS does not generate a state 2-design. Namely,
\begin{align}
\int\limits_{S^{2}}^{} | \theta,\phi \rangle \langle  \theta,\phi |^{\otimes t} d \Omega  \neq \gamma(t,d) \text{ for } t \geq 2.
\end{align}
As a consequence, SCS are not enough to perform tasks such as classical shadow tomography, which relies on 2- or 3-designs. Namely, to perform shadow tomography in these systems, one necessarily needs "spin squeezing," the equivalent of optical squeezing in bosonic systems.

An interesting contrast can be made here with respect to the Heisenberg-Weyl coherent states which also \textit{do not} form a 2-design. In the HW case, the lack of 2-design is simply a consequence of the \textit{infinite-dimensional} Hilbert space. The irrep structure breaks down into the same number of components as \(U(d)\) but it is the lack of convergence of the corresponding integral that prevents it from forming a \(2\)-design; see the discussion in this excellent reference \cite{blume-kohout_curious_2014}. In contrast, the failure of SCS to form a 2-design is not due to the infinite-dimensional convergence issues; rather the consequence is representation-theoretic. This example further highlights the fact that, while there are deep analogies between the spin and optical coherent states, via mappings such as Holstein-Primakoff, they are not identical in every way. The infinite-dimensional nature of bosonic systems manifests differently from the finite-dimensional nature of high-spin systems. In this case this resulted in two different ways to fail at being a unitary/state 2-design. There are also other ways in which spin and optical coherent states are different, the simplest being perhaps the fact that there exist two \textit{orthogonal} spin coherent states (corresponding to any pair of antipodal points on the Bloch sphere), while there are no two optical coherent states that are orthogonal (this only happens approximately as we send the Gaussian wavepackets asymptotically away from each other).

\subsection{Spin-GKP states and state 2-designs}
GKP codes are one of the leading proposals to perform quantum error correction with bosonic and high-dimensional qudit systems \cite{gottesman_encoding_2001}. Ref. \cite{conrad2024continuousvariabledesignsdesignbasedshadow} shows that the ensemble of all n-mode GKP stabilizer states over all symplectic lattices forms a "rigged 2-design". The ensemble is denoted as \(X_n := \{ | \Lambda; \vec{\alpha} \rangle = D(\vec{\alpha}) | \Lambda \rangle, \Lambda \in Y_n, \vec{\alpha} \in P(\Lambda) \}\). Here \(| \Lambda \rangle\) is an equal superposition of optical coherent states at all points in the symplectic self-dual lattice (the GKP code word that protects against displacement errors). This state is displaced as \(D(\vec{\alpha})\) via the set of all vectors in the lattice's fundamental domain, denoted as \(P(\Lambda)\). Recall that each lattice is defined by a symplectic matrix and the space \(Y_n = Sp(2n,\mathbb{Z}) \backslash Sp(2n,\mathbb{R})\). This compact space is endowed with a Haar measure that is then used to perform (uniform) averages over this set. 

The fact that spin-GKP codewords also form a state 2-design follows from combining the fact that the bosonic GKP codewords form a rigged 2-design and two other ideas. First, we use the \cref{thm:weighted-design-from-uniform} to project the bosonic lattice states onto a finite-dimensional Hilbert space, similar to the construction of the weighted state 2-design before. Then, we use the Holstein-Primakoff transformation to map these truncated "bosonic" degrees of freedom onto high-dimensional spins \cite{Holstein1940}.

The projection aspect is straightforward, by identifying \(| \widetilde{\Lambda; \vec{\alpha}} \rangle = \frac{P | \Lambda; \vec{\alpha} \rangle}{\|P| \Lambda; \vec{\alpha} \rangle\|}\) and similarly the weight function as in \cref{thm:weighted-design-from-uniform} for \(t=2\) (since we are interested in the 2-design case). It is then easy to show that averaging these projected states is equivalent to sandwiching the symmetric projector by the projector \(P\), which results into a symmetric projector onto a \(d\)-dimensional subspace.

The Holstein-Primakoff transformation is simply the map:
\begin{align}
\hat{S}_{+} \mapsto \sqrt{2 S} \hat{a}^{\dagger}, \quad \hat{S}_{-} \mapsto \sqrt{2 S} \hat{a}, \quad \hat{S}_{z} \mapsto \hat{N}-S
\end{align}
where \(N\) satisfies
\begin{align}
\left[\hat{N}, \hat{a}^{\dagger}\right]=\hat{a}^{\dagger}, \quad[\hat{N}, \hat{a}]=-\hat{a}, \quad\left[\hat{a}, \hat{a}^{\dagger}\right]=1-\hat{N} / S, \quad \hat{N} \rightarrow \hat{a}^{\dagger} \hat{a}.
\end{align}
Then one can see that in the limit \(S \gg 1\) this converges to the Heisenberg-Weyl algebra, i.e., the canonical commutation relations can be recovered \( \left[ a,a^{\dagger} \right] \approx 1\). Together, these suffice to show that spin-GKP codewords form a state 2-design, when averaging over all lattice states, and each lattice is displaced by all vectors in its fundamental domain.

\subsection{Approximate designs from native qudit gates}
The usual construction of unitary designs on multiqubit systems involves generating them from ``smaller'' amounts of randomness. This is captured in the brickwork architecture of random circuits following the seminal work of \cite{brandao_local_2016}. Here we are interested in a slightly different question: how do we generate unitary designs on a \textit{single} high-dimensional qudit, using \textit{native} gates such as SNAPs and displacements? The first thing to notice is that random circuits generated from a universal gate set such as the \(SU(4)\) gates or SNAP-displacements can generate an \textit{arbitrary} order of \(t\)-designs \cite{low_pseudo-randomness_2010}. The real question is, what is the \textit{minimal} depth needed for this? Or for the single qudit case, what is the minimal \textit{length} of such a circuit?

Fortunately, the work of Nakata \textit{et. al.} \cite{nakata_efficient_2017} already has a construction for this. The idea is to generate random unitaries in two different, \textit{mutually unbiased} bases. This randomizes the state of any system and can generate \(t\)-designs as a result, of arbitrary order (depending on the circuit length). Given two MUBs, let's call them \(\mathbb{B}\) and \(\mathbb{B}'\), such as the eigenbases of \(\hat{x}, \hat{p}\) for continuous-variable systems or the eigenbases of \(Z,X\) for qubits (does this work for qudits as well?), we want to understand what circuit lengths are sufficient to obtain a \(t\)-design. Using Corollary 3 of \cite{nakata_efficient_2017} we have that if \(d = \Omega (t^{2} (t!)^{2})\) (that is \(d \geq ct^{2} (t!)^{2}\), where \(c\) can be arbitrarily large) and \(\ell \geq \frac{1}{\log d-2 \log (t !)}[t \log d+\log (1 / \epsilon)]\) implies that the random unitary,
\begin{align}
D[\ell]:=D_{\ell}^{E} D_{\ell}^{F} D_{\ell-1}^{E} D_{\ell-1}^{F} \ldots D_{1}^{E} D_{1}^{F} D_{0}^{E},
\end{align}
generates an \(\epsilon\)-approximate unitary \(t\)-design. Here \(D^{E}\) is a random unitary in a fixed basis (say the computational basis) such as SNAP gates and \(D^{F}\) is a random unitary in a MUB. Note that \(d = \Omega (t^{2} (t!)^{2})\) is a technical requirement for the proof and it remains unclear if this can be relaxed. However, the numerics we discuss in [...] shows that this works for small \(d\) as well. One way to generate such a design from native gates is to apply random SNAP gates and utilize the quantum Fourier transform (QFT) to change to a MUB.

Namely, let \(U = \sum_{j=1}^{d} e^{i \phi_{j}} | j \rangle \langle  j |\) be a random diagonal unitary in the Fock basis. Then, if \(U_{Q}\)  is the QFT, we have \(U_{Q} U U^{\dagger}_{Q} = \sum_{j=1}^{d} e^{i \phi_{j}} | \widetilde{j} \rangle \langle  \widetilde{j} |\) is a random unitary, diagonal in the MUB. Therefore one can construct the \(t\)-designs as \(D[\ell]:=D_{\ell}^{E} (U_QD_{\ell}^{E} U^{\dagger}_{Q}) D_{\ell-1}^{E} (U_QD_{\ell-1}^{E} U^{\dagger}_{Q}) \ldots D_{1}^{E} (U_QD_{1}^{E}U^{\dagger}_{Q}) D_{0}^{E}\) instead. As a result, \(2\ell+1\) random unitaries in a fixed basis along with \(2 \ell\) QFTs are sufficient to generate a single qudit \(\epsilon\)-approximate \(t\)-design.

The numerical result of \cite{job_efficient_2023} tells us that \(2 \ell +1\) random SNAP gates along with \(2 \ell \times O(d)\) SNAP + displacement gates suffice for this. As a result, we have that \(O(\ell d)\) SNAP-displacement gates suffice to generate a unitary \(t\)-design. As an example, if we want a \(2\)-design for \(d = O(1)\), this gives us \(\ell \gtrapprox 2 + \log_{d}(1/\epsilon)\). For \(\epsilon = 1/d^k\) we have \(\ell \gtrapprox 2 + k\). Therefore, only \(O(d)\) SNAP-displacement gates suffice to generate a unitary \(2\)-design. Notice that the majority of the gate cost comes from generating the QFT. While these bounds can obviously be improved, we believe they will be instructive for thinking about generating designs from native gate sets.

\subsection{Cavity-QED qudits}
Superconducting radio-frequency (SRF) cavities offer an interesting platform for hosting qudits: the cavity has long coherence times (typically in milliseconds) and can be easily integrated with nonlinear elements for universal control \cite{heeres_cavity_2015,Eickbusch2022,kim2025ultracoherentsuperconductingcavitybasedmultiqudit}. The qudit is typically encoded in the first \(d\)-levels of a bosonic mode and control is implemented via "naturally occurring" interactions such as selective number-dependent arbitrary phase (SNAP) gates \cite{krastanov_universal_2015}, displacements, two-level rotations (Givens rotations), echoed conditional displacement (ECD) gates \cite{Eickbusch2022}, sideband interactions \cite{Liu2021,huang2025fastsidebandcontrolweakly}, among others.

For our discussion below, we focus on the simpler case of SNAP and displacement gates or SNAP and Givens rotations as our universal gate set (both sets are independently universal). Since they are universal gate sets, they can form arbitrary orders of unitary \(t\)-designs in the \(d\)-dimensional computational subspace. Note that this only works for \(d<\infty\), since convergence on the full unitary group is subtle in the infinite-dimensional case; see some work on rigged t-designs \cite{iosue_continuous-variable_2024}. How such random unitaries can be used to assess the performance of noisy qudits was discussed in Ref. \cite{bornman2024benchmarkingalgorithmicreachhighq}.

\section{Character Randomized Benchmarking}
\label{sec:qudit character RB}

We now discuss character randomized benchmarking (RB) \cite{helsen2019new, claes_character_2021, helsen_general_2022}. 
Let $\mathcal{H} = \mathbb{C}^d$ be a state space, with corresponding unitaries $U(d)$. Let $G\subseteq U(d)$ be a finite or compact subgroup. 
We consider the problem of approximating the average fidelity of gates in $U(d)$ using only gates in $G$. 
We will need to consider the representations of $G$ and $U(d)$, acting by conjugation on 
$\mathcal{L}(\mathcal{H})$. 
For convenience, we often use the canonical isomorphism $\mathcal{L}(\mathcal{H})\cong \mathcal{H}\otimes \mathcal{H}^*$, with $U\in U(d)$ acting by $U\otimes U^*$. 

Recall that we have a decomposition
\begin{align}
    \mathcal{L}(\mathcal{H}) = \sum_i V_i^{\oplus m_i},
\end{align}
where the $V_i$ are inequivalent $G$-irreducibles and the $m_i>0$ are their multiplicities. 
(In general one may consider irreps with respect to a subgroup of $G$ \cite{claes_character_2021}; we neglect this for ease of exposition.) 
Earlier works on RB assume that $G$ is a $2$-design, which in particular implies that $\mathcal{L}(\mathcal{H})\cong \mathcal{H}\otimes \mathcal{H}^* \cong \mathfrak{sl}(\mathcal{H})\oplus\mathbb{C}I$ (see the proof of Theorem B.1 in \cite{Saied2024}). 
In particular, when $G$ is a $2$-design there are two irreps, each of multiplicity $1$, with one being the trivial representation. 
In the more general case of character RB, we will need to consider the \emph{isotypic components} $W_i = V_i^{\oplus m_i}$, with projectors $P_i$. 
We need, for each isotypic component $W_i$, a state $\rho_i$ and measurement projector $M_i$ with $|\Tr(M_i P_i \rho_i)|\neq 0$, preferably as large as possible. 
In the examples considered below, we will in fact choose a single state $\rho$ sufficient for all isotypic components, and will similarly measure only a single observable (with the different outcomes potentially relevant to different isotypic components). 
The RB experiment then involves the following for each $i$: preparing the state $\rho_i$, randomly choosing a sequence $U_0, U_1, \dots, U_{N}\in G$ (according to the uniform or Haar measure), applying the gates $U_1 U_0, U_2, U_3, \dots, U_{N}, U_1^\dagger\cdots U_N^\dagger$ in sequence (with the first and last compiled as single gates), and measuring an appropriate observable to project onto $M_i$. 
Interpreting this data appropriately then gives the RB parameters. 
We will not go into detail here, referring the reader to \cite{helsen2019new, claes_character_2021, helsen_general_2022}, but we briefly give some intuition for why these methods should work. 

During the RB process, the relevant channels correspond to gates in $G$ and a fixed gate error channel $\Lambda$ (in addition to state preparation and measurement errors). 
Through averaging over random choices of $U_1, \dots, U_N$, we are able to replace most appearances of $\Lambda$ with its $G$-twirl $\Lambda_G = \frac{1}{|G|}\sum_{U\in G}U \Lambda U^\dagger$, which notably commutes with the action of $G$. 
In particular, both elements of $G$ and $G$-twirls preserve the isotypic components $W_i$. 
Then character RB essentially reduces to studying each isotypic component individually: if one starts with a state $\rho_i\in W_i$, applies only gates in $G$ and $G$-twirls of quantum channels, then projects onto an observable in $W_i$, the entire computation remains within that isotypic component. In practice, SPAM errors may cause us to leave the isotypic component; further, as discussed above, we do not need to choose $\rho_i$ and $M_i$ in $W_i$. 
This is the role of the "extra" random gate $U_0$ discussed above: as shown in \cite{helsen2019new, claes_character_2021}, one may take a cleverly weighted average over $U_0\in G$ (with complex-valued weights coming from the character of $V_i$) to insert the projector $P_i$ into the relevant expressions, forcing the calculation back into $W_i$. 
The results from each irrep can then be combined to approximate the overall average fidelity. 

The situation is particularly simple when all multiplicities $m_i=1$, because Schur's lemma then implies that the $G$-twirls act by a scalar, and we may perform randomized benchmarking by fitting a single exponential decay for each irrep. This is the special case considered in Ref. \cite{helsen2019new}. 
More generally, the case with nontrivial multiplicities is discussed in \cite{claes_character_2021,helsen_general_2022} and involves fitting a sum of multiple exponential decays. 

In the setting of Proposition~\ref{prop:multiplicity 1}, we find all the irreps to have multiplicity $1$, leading to character randomized benchmarking using single exponential decays as in Ref. \cite{helsen2019new}. 
We will consider the special case $G\cong SU(2)$ in Section~\ref{sec:character rb su2}, and yet another multiplicity-free case in Section~\ref{sec:character clifford}. 
In order to enable character RB in these settings, we will (1) give the decomposition of $\mathcal{L}(\mathcal{H})$ into irreps; (2) exhibit a useful initial state $\rho$ overlapping nontrivially with all irreps; and (3) give the appropriate observables with measurement projectors overlapping nontrivially with the various $P_i \rho$. 

\subsection{Character RB for \(SU(2)\)}\label{sec:character rb su2}
We will now discuss the application of character randomized benchmarking to $SU(2)$. We also refer the reader to Ref. \cite{fan2024randomizedbenchmarkingsyntheticquantum} for a systematic discussion of "synthetic RB" for the \(SU(2)\) case. In particular, following the previous section, we will discuss the decomposition of $\mathcal{H}_S\otimes\mathcal{H}_S^*$ into (multiplicity-free) $SU(2)$-irreps, then give the required initial states and measurements. 
In fact, for ease of implementation, we will give a single state and single measurement that suffice to benchmark all irreps. 

Note that we have an equivalence of $SU(2)$ representations $\mathcal{H}_S^*\rightarrow \mathcal{H}_S$, given by $\bra{S,m}\mapsto \ket{S,-m}$. 
Then we have 
\begin{align}
    \mathcal{L}(\mathcal{H}_S) \cong  \mathcal{H}_S\otimes \mathcal{H}_S^*\cong \mathcal{H}_S\otimes \mathcal{H}_S,
\end{align}
with the isomorphism between the latter two given by 
\begin{align}\label{eq:irrep tensor iso}
    \ket{S,m_1}\ket{S,m_2}\mapsto \ket{S,m_1}\bra{S,-m_2}.
\end{align}
We will freely move between these descriptions of $\mathcal{L}(\mathcal{H}_S)$. 

Note that $\mathcal{H}_S$ has weights $2(S-m)$, where $0\leq m\leq 2S$ is an integer. 
(Recall that the \emph{weights} are double the spins, and correspond to the eigenvalues of $Z\in \mathfrak{su}(2)$.) 
Then the tensor product has weights 
\[2(S-m_1) + 2(S-m_2) = 2(2S - (m_1 + m_2)).\]
In particular, only even integer weights arise in $\mathcal{H}_S\otimes \mathcal{H}_S\cong \mathcal{H}_S\otimes \mathcal{H}_S^*$. 

More precisely, from Appendix~\ref{sec:character appendix}, we have
\begin{align}
    \mathcal{H}_S \otimes \mathcal{H}_S^* \cong \bigoplus_{\ell=0}^{2S}V_{(2S-\ell, \ell-2S)} \cong \bigoplus_{\ell=0}^{2S}\mathcal{H}_{2S-\ell}. 
\end{align}
Then the highest weights of the relevant irreps have the form $2J$ for $J$ an integer satisfying $0\leq J\leq 2S$. 
In particular, each irrep has multiplicity $1$. 

Following the above, we now identify a density matrix $\rho$ that overlaps nontrivially with each irrep. It turns out we can simply take $\rho = \ket{S,S}\bra{S,S}$, the density matrix corresponding to the highest weight vector. Recall that under the equivalence \eqref{eq:irrep tensor iso}, we may view $\ket{S,S}\bra{S,S}$ as the tensor $\ket{S,S}\ket{S,-S}$. Then the overlaps of basis vectors in the irreps $V_i$ with $\ket{S,S}\ket{S,-S}$ may be calculated using the Clebsch-Gordan coefficients. 
From the explicit formula of \cite{bohm2013quantum}, given in (2.41) on page 172, we directly see that the relevant Clebsch-Gordan coefficients are nonzero. (In fact, the formula, which is generally a large sum, reduces to only a single nonzero term.) 
We may explicitly calculate that for the irrep $\mathcal{H}_{J}$ with highest weight $2J$ (where $J$ is an integer satisfying $0\leq J\leq 2S$), the unique unit vector with weight $0$ has the following inner product with $\ket{S,S}\ket{S,-S}$: 
\begin{align}
    \left(  \dfrac{(2J+1)(2S)!^2}{(2S+J+1)!(2S-J)!} \right)^{1/2}.
\end{align}
We conclude that $\ket{S,S}\ket{S,-S}$ overlaps nontrivially with each irrep, and therefore the same holds for $\rho = \ket{S,S}\bra{S,S}$. 
We may also take our measurement projector $M = \rho$, so that, for $P_J$ the projector onto the copy of $\mathcal{H}_{J}$, 
\[\Tr(MP_J\rho) = \bra{S,S}P_J\ket{S,S}\neq 0\]
as desired.

A similar argument allows us to take $M = \rho = \ket{S,-S}\bra{S,-S}$, the density matrix corresponding to the \emph{lowest} weight vector. 
In fact, we may instead use $M = \ket{S,S}\bra{S,S} + \ket{S,-S}\bra{S,-S},$ which has nontrivial overlap with each irrep. This observable corresponds to measuring in the standard basis $\ket{S,m}$ and keeping only those terms that project onto $\ket{S,S}$ and $\ket{S,-S}$. 
Then we may perform character RB for $SU(2)$ using only standard basis states and measurements. 

\begin{proposition}
    For $\mathcal{H}_S$, the group of gates corresponding to $SU(2)$ admits multiplicity-free character RB, using a single state $\rho = \ket{S,S}\bra{S,S}$ and a single measurement projector $M = \ket{S,S}\bra{S,S} + \ket{S,-S}\bra{S,-S}$. 
    RB is performed by fitting $2S+1$ single exponential decays, one for each irrep. 
\end{proposition}

\subsection{Character RB for the single-qudit Clifford group in arbitrary dimension}\label{sec:character clifford}

We now consider character randomized benchmarking for $Cl_d$, the Clifford group of a single $d$-dimensional qudit in arbitrary dimension $d$. 
Recall that this involves decomposing the space $\mathcal{L}(\mathcal{H})$ into irreps, where $\mathcal{H} = \mathbb{C}^d$. 
We do this explicitly: 
\begin{theorem}
    Under the conjugation representation of $Cl_d$, we have
    \begin{align*}
        \mathcal{L}(\mathbb{C}^d) \cong \bigoplus_{1\leq r\leq d,\,\, r|d} W_r,
    \end{align*}
    where 
    \begin{align*}
        W_r = \textnormal{span}_{\mathbb{C}}\{X^a Z^b: \gcd(a,b,d)=r\}
    \end{align*}
    are irreps for $Cl_d$ of multiplicity $1$, and the dimension of $W_r$ is the number of pairs of integers $(a,b)$ satisfying $0\leq a,b< d$ and $\gcd(a,b,d)=r$. 
    In particular, $Cl_d$ is a $2$-design if and only if $d$ is prime. 
\end{theorem}
\begin{proof}
    Recall from the introduction to Section~\ref{sec:qudits} that we may identify the conjugation action of $Cl_d$ on $\mathcal{P}_d$ with the multiplication action of $Sp_2(\mathbb{Z}_d)$ on vectors $(a,b)\in \mathbb{Z}_d^2$ (written as ordered pairs rather than column vectors for convenience). We will first establish the orbits $V_r$ of the symplectic action. We then show that the corresponding subrepresentations $W_r$ are irreducible and inequivalent, so they have multiplicity $1$. 
    
    Note that
    \begin{align}
        \mathbb{Z}_d^2 = \bigcup_{\substack{r|d \\ 1\leq r\leq d}}V_r\textnormal{, where }V_r = \{(a,b): 0\leq a,b<d, \gcd(a,b,d)=r\},
    \end{align}
    and the $V_r$ are pairwise disjoint. 
    Factoring out the greatest common divisor, we may write
    \begin{align}\label{eq:orbit characterization}
        V_r = \{(ra',rb'): 0\leq a',b'<d/r, \gcd(a',b',d/r)=1\}.
    \end{align}
    By linearity and invertibility, the $V_r$ are clearly preserved by the action of $Cl_d$; further, we may show the action on each $V_r$ is \emph{transitive}, since $(r,0)$ may be taken to arbitrary $(ra',rb')$ by some symplectic matrix with first column $(a',b')^T$. 
    Then the $V_r$ are the full orbits of the $Cl_d$ action. 

    Of course, the elements of $V_r$ index the Pauli basis for $W_r$ (defined in the theorem statement), with $(a,b)\mapsto X^a Z^b$. 
    We have thus shown that, ignoring phases, $Cl_d$ transitively permutes the Pauli basis of $W_r$. 
    Further, the Pauli subgroup $\mathcal{P}_d\subset Cl_d$ acts by the appropriate phases (powers of a primitive $d$th root of unity $\omega$). 
    We may use these two facts to explicitly show that the $W_r$ are irreducible. 
    One may start with an arbitrary linear combination $v=\sum_{(a,b)\in V_r} c_{a,b} X^a Z^b$ and inductively apply Clifford group elements to reduce the number of nonzero terms, eventually taking arbitrary $v\in V_r$ to a fixed basis element, say $X^r$. 

    We now show the $W_r$ are inequivalent, separately considering the cases in which $d$ is odd and even. 
    The idea of the proof is to exhibit a Clifford group element $L$ that can be used to tell the various $W_r$ apart. 
    Specifically, we will use the fact that an equivalence of irreps must preserve the order of $L$. 

    First, if $d$ is odd, we consider the diagonal unitary operator $L=(\delta_{ij}\omega^{j(j+1)/2})_{ij}$. Clearly $L$ commutes with $Z$, and further one can check $LXL^\dagger = \omega XZ$, so $L\in Cl_d$. 
    (This is an analogue of the $S=\sqrt{Z}$ gate for qubits.) 
    This conjugation relation implies that    
    $L$ has order $d/r$ when acting by conjugation on the irrep $W_r$. 
    This prevents any possible equivalences between the irreps. 

    If $d$ is even, we let $\zeta=\exp(\pi i/d)$, so that $\zeta^2 = \omega$, and let $L=(\delta_{ij}\zeta^{j^2})_{ij}$. 
    We similarly have that $L\in Cl_d$, with $LXL^\dagger = \zeta XZ$. 
    In this case, the phases are more complicated. We may calculate that on $V_r$, 
    $L$ has order $d/r$ if $r$ is even and $2d/r$ if $r$ is odd. 
    Like above, this prevents all possible equivalences. 
    To make this clear, suppose for the sake of contradiction that we have an equivalence $V_{r'}\rightarrow V_r$, where $r'<r$. 
    Since $d/r < d/r'$, and equivalence must preserve the order of $L$, this would require $L$ to have order $2d/r$ on $V_r$ (forcing $r$ to be odd), order $d/r'$ on $V_{r'}$ (forcing $r'$ to be even), and further $2d/r=d/r'$. This equation may be rewritten as $r=2r'$. But this implies that the odd number $r$ is divisible by $2$, a contradiction. 
    We conclude that all irreps $W_r$ are inequivalent. 
    
    The statement about $2$-designs follows from Proposition~\ref{prop:designs-equivalent-defns}. 
\end{proof}

Given this result, we see that single-qudit character RB is very tractable: the irreps are multiplicity-free, so there are only single exponential decays. Further, the number of decays we need to fit is equal to the number of divisors of $d$, which is trivially $O(\sqrt{d})$. (With much tighter bounds known, see \cite{hardy1979introduction}.) 

To fully enable character RB, we must also specify what initial states and measurements to perform. 
Fortunately, for any $s$, we may always take $M = \rho = \ket{s}\bra{s}$, since $Z^r\in W_r$ and 
\begin{align}
    \Tr(MZ^r\rho) = \bra{s}Z^r \ket{s} = \omega^{rs}.
\end{align}
In particular, we can prepare a fixed computational basis state, say $\rho = \ket{0}\bra{0}$, perform computational basis measurements, and use the projector $M = \ket{0}\bra{0}$. 

\begin{proposition}
    The single-qudit Clifford group in arbitrary dimension $d$ admits multiplicity-free character RB, using a single state $\rho=\ket{0}\bra{0}$ and a single measurement projector $M = \ket{0}\bra{0}$. 
    RB is performed by fitting $\phi(d)$ single exponential decays, one for each irrep, where $\phi(d)$ is the number of divisors of $d$. 
\end{proposition}

\section{Fractional designs}
\label{sec:fractional designs section}

In this section we introduce a notion of "fractional \(t\)-designs," that provide a natural generalization of \(t\)-designs beyond integer order. These are inspired from the theory of \textit{fractional calculus} \cite{jones_riemann-liouville_1991}. Recall that a set of \(N\) pure states form a (uniform) \(t\)-design if and only if \cite{scott_optimizing_2008} (see also \cref{eq:welch-test-weighted-design}),
\begin{align}
\frac{1}{N^{2}} \sum\limits_{j,k=1}^{N} \left| \left\langle \psi_{j} | \psi_{k} \right\rangle \right|^{2t} = \frac{1}{\binom{d+t-1}{t}}.
\end{align}
The RHS is obtained by plugging in Haar random states and using the symmetric projector on \(\mathcal{H}^{\otimes t}\) (along with Schur's lemma). Our goal is to generalize this formula to a non-integer \(t\). Notice that such a generalization does not make sense (in any straightforward way) from the point of view of the tensor product, i.e., we are \textit{not} going to think about the quantity, 
\begin{align}
\mathbb{E}_{\psi \sim \mathrm{Haar}} \quad | \psi \rangle \langle  \psi | ^{\otimes t}.
\end{align}
Instead, we notice that the inner product on the LHS can be evaluated for any \(t \in \mathbb{R}_{\geq 0}\). This will be way in which we generalize the definition of a \(t\)-design. Notice that the RHS can be generalized via the use of the Gamma functions,
\begin{align}
\Gamma(z) = \int\limits_{0}^{\infty} t^{z-1} e^{-t} dt , \quad \mathfrak{Re}(z) > 0.
\end{align}
Moreover, for a positive integer $n$ we have \(\Gamma(n) = (n-1)!
\). Importantly, the Gamma function allows us to generalize the notion of a factorial in the RHS to any \(t \in \mathbb{N}\)
\begin{align}
\frac{1}{\binom{d+t-1}{t}} = \frac{\Gamma(t+1) \Gamma(d)}{\Gamma(t+d)}.
\end{align}
This will be the starting point of the notion of a fractional t-design for us.

\begin{definition}[Fractional state designs]
A set of quantum states in \(\mathcal{H} \cong \mathbb{C}^{d}\) form a fractional \(t\)-design if and only if,
\begin{align}
\frac{1}{N^{2}} \sum\limits_{j,k=1}^{N} \left| \left\langle \psi_{j} | \psi_{k} \right\rangle \right|^{2t} = \frac{\Gamma(t+1) \Gamma(d)}{\Gamma(t+d)} = t B(t,d).
\label{eq:fractional-welch}
\end{align}
\end{definition}

To argue that this is a meaningful generalization, we first show that the RHS holds true for \(t \in \mathbb{R}_{\geq 0}\) for Haar random states (the \(t=0\) case is trivial so we only focus on \(t \in \mathbb{R}_{>0}\)). Namely, we want to compute,
\begin{align}
\int_{\psi, \phi \sim \mathrm{Haar}} \left| \left\langle \psi | \phi \right\rangle \right|^{2t} = \int \left| \left\langle \psi | \phi \right\rangle \right|^{2t} d\mu(\psi) d\mu(\phi).
\end{align}
Here, \(d\mu(\psi)\) denotes the unique, unitarily-invariant measure on the complex projective Hilbert space. Using the unitary invariance allows us to translate one of the states, \(| \phi \rangle \mapsto | 0 \rangle\) without changing the result of the integral above. Therefore, we only need to estimate,
\begin{align}
\int \left| \left\langle 0 | \psi \right\rangle \right|^{2t} d\mu(\psi).
\end{align}
As an aside, note that for \(t=1\) this is akin to estimating the probability of observing the bitstring corresponding to \(| 0 \rangle\) when measuring a Haar random state. To perform the integral, we will parametrize pure quantum states as \(| \psi \rangle = \sum\limits_{j=1}^{d} \psi_{j} | j \rangle\) where, \(\psi_{j} = x_{j} + i y_{j}\). Then, the Haar measure over normalized pure states can be parametrized in spherical coordinates as \cite{singh_average_2016}
\begin{align}
d\mu(\psi) = \frac{\Gamma(d)}{(2 \pi)^{d}} \; \delta\left(1-\sum_{j=1}^{d} r_{j}\right) \prod_{j=1}^{d} d r_{j} d \theta_{j},
\end{align}
where \(x_{j} = \sqrt{r_{j}} \sin \left( \theta_{j} \right)\) and \(y_{j} = \sqrt{r_{j}} \cos \left( \theta_{j} \right)\) with \(r_{j} \in \left[ 0,1 \right] \) and \(\theta_{j} \in [0,2 \pi]\). The \(\delta\)-function ensures the normalization of these pure states. Therefore,
\begin{align}
\int \left| \left\langle 0 | \psi \right\rangle \right|^{2t} d\mu(\psi) = \int \left| r_{1} \right|^{t} d\mu(\psi),
\end{align}
where we let \(| 0 \rangle\) be the first element of the basis in which we are expanding above, hence we obtain the term \(r_{1}\). Plugging in the measure from above and using the fact that \(\left| r_{1} \right| = r_{1}\) we have,
\begin{align}
\mathrm{LHS} = \Gamma(d) \int \left( r_{1} \right)^{t} \delta\left(1-\sum_{j=1}^{d} r_{j}\right) \prod_{j=1}^{d} d r_{j},
\end{align}
where we have already integrated over the \(\theta_{j}\)'s to obtain a factor of \(\left( 2 \pi \right)^{d}\) cancelling the prefactor in the measure above. Then,
\begin{align}
\mathrm{LHS} &= \Gamma(d) \int\limits_{0}^{1} dr_{1} \left( r_{1} \right)^{t} \int\limits_{0}^{\infty} \delta \left( \left( 1-r_{1} \right)  - \sum\limits_{j=2}^{d} r_{j} \right) \prod_{j=2}^{d} d r_{j} \\
&= \frac{\Gamma(d)}{\Gamma(d-1)} \int\limits_{0}^{1} dr_{1} \left( r_{1} \right)^{t} \left( 1-r_{1} \right)^{d-2}.
\end{align}
Now we just need to recall the Beta function,
\begin{align}
B(\alpha, \beta):=\int_{0}^{1} r^{\alpha-1}(1-r)^{\beta-1} d r=\frac{\Gamma(\alpha) \Gamma(\beta)}{\Gamma(\alpha+\beta)}.
\end{align}
Therefore, the integral above is equal to \(B(t+1,d-1)\), which in terms of the \(\Gamma\) function gives us,
\begin{align}
B(t+1,d-1) = \frac{\Gamma(t+1) \Gamma(d-1)}{\Gamma(t+d)}.
\end{align}
Adding the prefactors from above we have (after cancelling the \(\Gamma(d-1)\)),
\begin{align}
\mathrm{LHS} = \frac{\Gamma(t+1) \Gamma(d)}{\Gamma(t+d)}.
\end{align}
Therefore, we have shown that Haar random states form fractional \(t\)-designs for any \(t >0\). Recall that the definition of the Gamma function requires \(\mathfrak{Re}(z) >0\) which amounts to \(t \in \mathbb{R}_{>0}\). The result can be generalized from states to unitaries as follows.

\begin{definition}[Fractional unitary designs]
A finite subgroup $\{U_i\}_i$ of the unitary group forms a fractional \(t\)-design if and only if the \emph{generalized frame potential} satisfies
\begin{align}
\frac{1}{N^{2}} \sum\limits_{j,k=1}^{N} \left| \mathrm{tr}(U_j^\dag U_k) \right|^{2t} = 
    \int |\mathrm{tr}(U)|^{2t}d\mu_U. 
\label{eq:fractional-frame-potential}
\end{align}
\end{definition}

\begin{figure}[!t]
\centering
\includegraphics[width=0.6\linewidth]{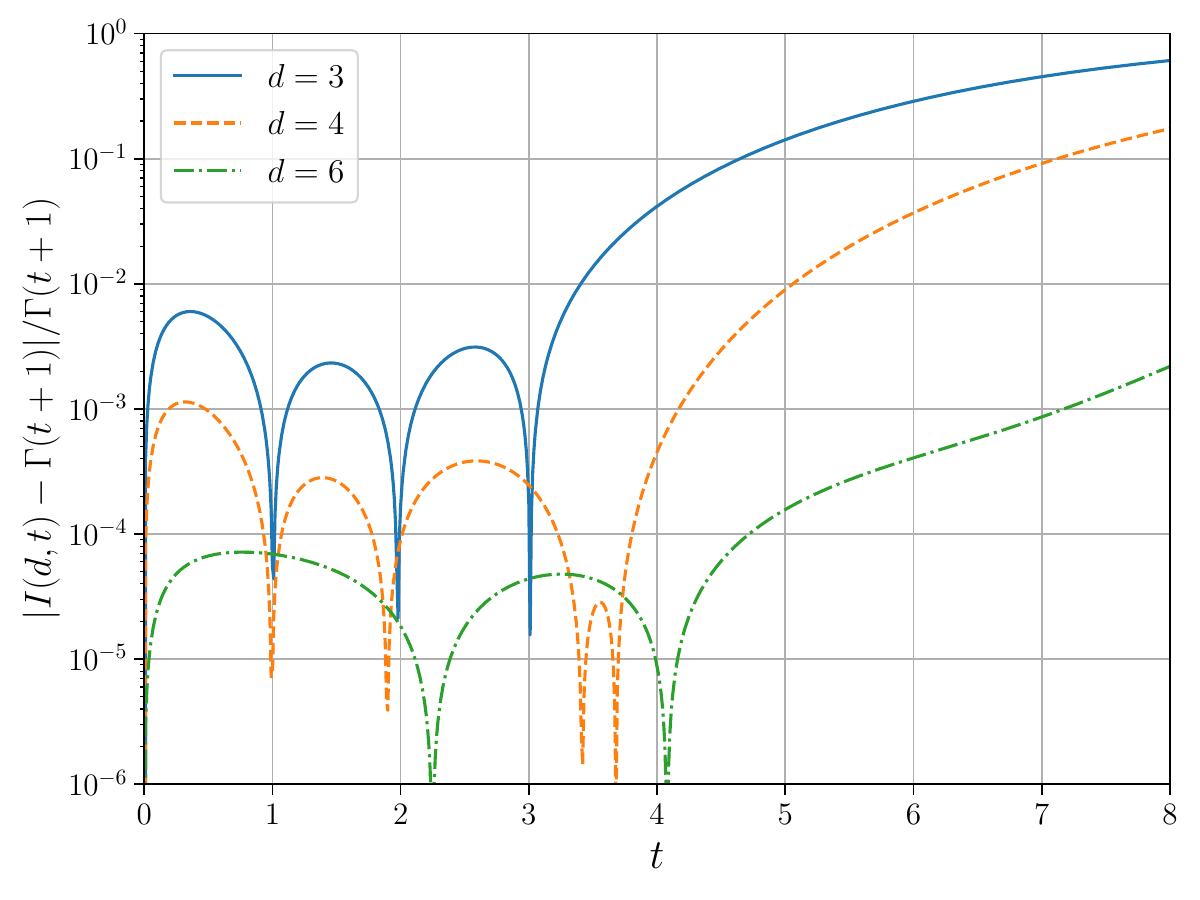}
\caption{Numerical evaluation of the relative error between the Haar integral $I(d,t):=\int_{\mathcal{U}(d)} |\mathrm{tr}(U)|^{2t} d\mu_U$ and the Gamma function $\Gamma(t+1)$, for $d=3,4,6$. This is computed by taking the trace of $10^8$ Haar random unitary samples (using \texttt{scipy}), and numerically estimating the integral. The random looking form of these curves is likely explained by finite sampling effects, however, the broad features are fairly transparent: We see in the regime $t \le d$, the approximation is fairly accurate (relative error less than 1\% in all cases). Also notice the error generally reduces increasing $d$, consistent with the expectation that in the asymptotic limit, $I$ does indeed converge to $\Gamma(t+1)$. We observe the error \emph{does not} appear to go to zero with the number of samples, suggesting, as per the main text, that the integral value is not exactly $\Gamma(t+1)$, even for $t\le d$.}
\label{fig:frame-error}
\end{figure}

If $d=2$, we have
\begin{align}
    \int |\mathrm{tr}(U)|^{2t}d\mu_U = \frac{\Gamma(2t+1)}{\Gamma(t+1)\Gamma(t+2)}.
\end{align}
To enable approximations for $d>2$, we also note that
\begin{align}
    \lim_{d\rightarrow\infty}\int |\mathrm{tr}(U)|^{2t}d\mu_U = \Gamma(t+1).
\end{align}

For proof of the $d=2$ case, we first note that the eigenvalue spacing statistics for Haar random unitaries are given by density $p(\Delta) = \frac{1}{\pi}(1-\cos \Delta)(1-\Delta/2\pi)$, where $\Delta\in[0,2\pi)$ is the angular difference of the two eigenvalues (this is shown in App.~\ref{sect:eigen-density}). The integral of interest can then be written as:
\begin{align}
    \int |\mathrm{tr}(U)|^{2t} d\mu_U = \int_0^{2\pi} |1+e^{i\Delta}|^{2t}p(\Delta)d\Delta = \frac{2^t}{\pi} \int_0^{2\pi} (1+\cos \Delta)^t (1-\cos \Delta)\left(1 - \frac{\Delta}{2\pi}\right) d\Delta.
    \label{eq:d_2_integral}
\end{align}
In App.~\ref{sect:integral_proof} we show this can be evaluated to give the desired result, valid for $\mathrm{Re}(t) > -1/2$.

The $d>2$ case is more complicated as discussed below, and we do not have a general form for the integral evaluation in \eqref{eq:fractional-frame-potential}. The limiting case can however be evaluated, noting that in the asymptotic limit  $\mathrm{tr}(U)$ is Gaussian distributed with mean zero and variance $1/2$ in the real and imaginary axes (covariances are 0) \cite{diaconis_eigenvalues_1994}:
\begin{align}
    \int |\mathrm{tr}(U)|^{2t} d\mu_U = \frac{1}{\pi} \int |z|^{2t} e^{-|z|^2}dz = 2 \int_0^\infty r^{2t} e^{-r^2} rdr = \int_0^\infty q^{t} e^{-q}dq = \Gamma(t+1).
    \label{eq:gamma_proof}
\end{align}
Numerically we observe $\Gamma(t+1)$ to be a reasonable approximation for the integral in the regime $t \le d$ (see Fig.~\ref{fig:frame-error}), and in numerics below we use this value in lieu of the exact expression.

We make an important comment here for the $d>2$ case. In the integer case, the frame potential is the number of permutations of $\{1, \dots, t\}$ with no increasing subsequence of length larger than $d$ \cite{zhu_clifford_2016}. For $t\le d$ that is clearly $t!$, and for $t>d$, it is strictly less than $t!$. The naive expectation for the generalization to non-integer $t$ would perhaps be that for $d\ge t$, the generalized frame potential takes on value $\Gamma(t+1)$. However, by the identity theorem of complex analysis, two holomorphic functions that are equal on a subdomain, are equal on the full domain of analyticity (in this case, the region of interest being $\mathrm{Re}(t)>0$). This would imply the Haar integral would be $\Gamma(t+1)$ for all $t>0$, which is a contradiction of the integer case for $t>d$. Therefore, the functions cannot be equal. We expect the correct generalization is based directly on the definition of permutations of $\{1, \dots, t\}$ using all integer values of $t$, which is highly non-trivial for $t > d$ (e.g. see p.~281 in \cite{GESSEL1990257} for the $d=3$ case). As mentioned above however, the approximation of $\Gamma(t+1)$ for $t\le d$ is fairly accurate based on numerical observations, and increasingly so for larger $d$.

\subsection{Numerics/discussion for fractional $t$-designs}
\begin{figure}[!ht]
    \centering
    \includegraphics[width=0.5\linewidth]{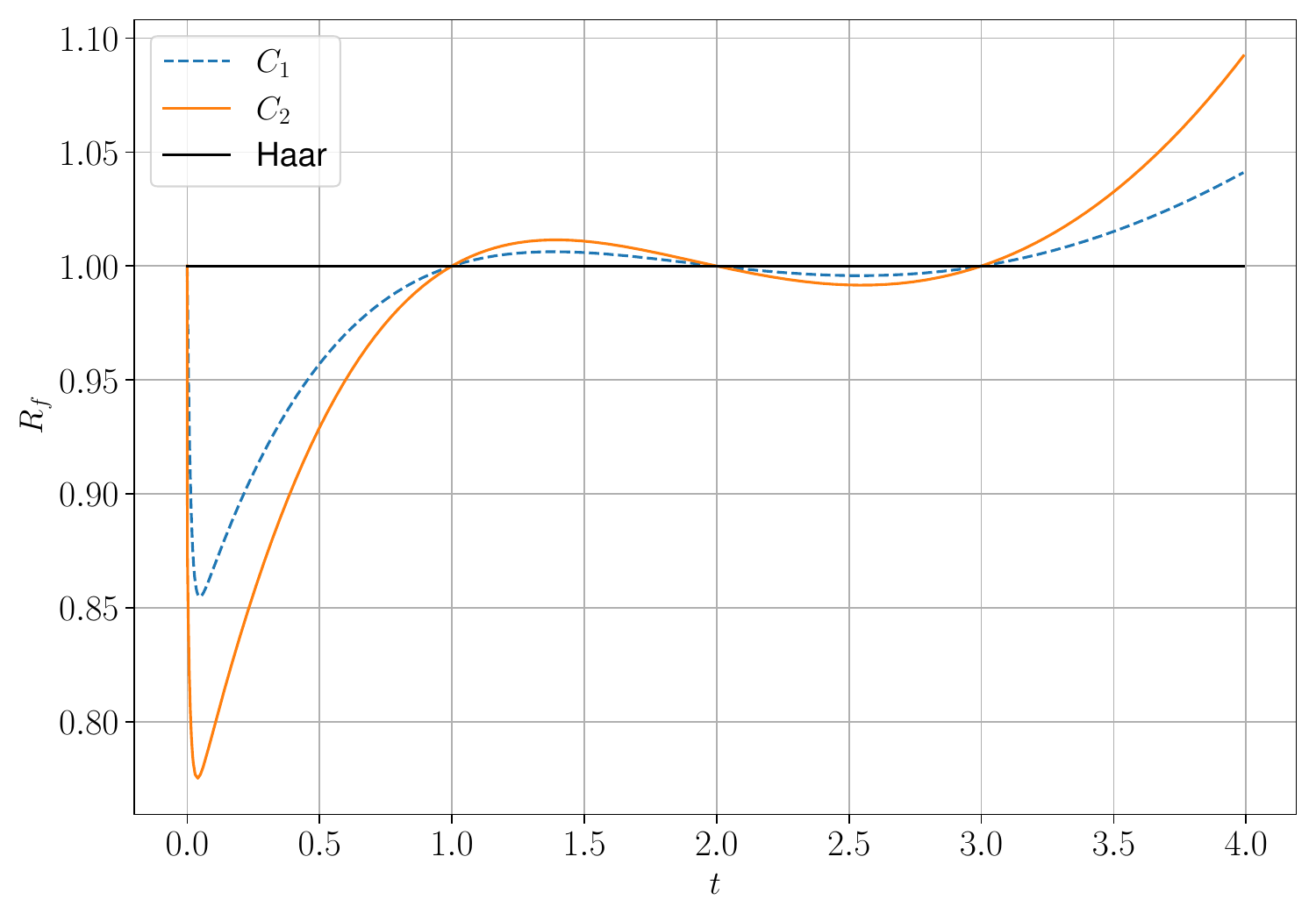}
    \caption{Fractional frame potential ratio for one and two qubit Clifford group.}
    \label{fig:clifford-fractional}
\end{figure}

Whilst we have been unable to find an operational group theoretic meaning for non-integer $t$-designs, they do exist, as we have observed numerically. First, to gain some intuition, in Fig.~\ref{fig:clifford-fractional} we plot the frame potential ratio $R_f$, defined as the frame potential at (non-integer) value $t$, divided by the Haar random result Eq.~\eqref{eq:fractional-frame-potential}, for the Clifford group on 1 and 2 qubits. We see it is only an integer design, for $t=0,1,2,3$.
Generalizing from integers to the reals can also provide us with more information about a group. Notice in Fig.~\ref{fig:clifford-fractional} that at $t=2$, the frame potential has a negative gradient. Since we know for the integers $R_f \ge 1$, by continuity, it implies the Clifford group must be a design for some (potentially non-integer) $t > 2$. 

Another group of interest is a qubit representation of $SL(2,\mathbb{F}_5)$, which is a 5-design. $SL(2,\mathbb{F}_5)$ is the group of unitary $2\times 2$ matrices of determinant 1 over $\mathbb{F}_5$. The group is isomorphic to the \emph{binary icosahedral} group, with generators $\langle r,s,t\rangle$, where $r^2=s^3=t^5=rst=-I$, and has order 120. A unitary two-dimensional representation can easily be constructed from these constraints directly. For example, the following are unitary generators:
\begin{align}
    r=\frac{1}{\sqrt{2}}\left(\begin{array}{cc}
        i x &  -y\\
        y & -i x
    \end{array}\right),\;\; s=\frac{1}{\sqrt{2}}\left(\begin{array}{cc}
       i x\, \omega^{-1} &  -y\,\omega\\
        y\, \omega^{-1} & -ix\, \omega
    \end{array}\right),\;\;t=\left(\begin{array}{cc}
        \omega  &  0\\
        0 & \omega^{-1}
    \end{array}\right),
\end{align}
where $x = \sqrt{1 + 1/\sqrt{5}}$, $y=\sqrt{1 - 1/\sqrt{5}}$, and $\omega=e^{2\pi i/10}$.  
We plot the generalized frame potential for this representation in Fig.~\ref{fig:5-design}.
This representation was also described in \cite{gross_evenly_2007, ketterer_entanglement_2020}; however, we believe there to be typographic mistakes in the provided generators for each reference, which we have corrected in App.~\ref{sec:5-design-construction}.

We note that $G=SL(2,\mathbb{F}_5)$ in fact has two inequivalent $2$-dimensional irreps. 
We will show that \emph{both} of these inequivalent irreps can be described using the above generators. 
Let $\rho_1: G\rightarrow U(2)$ be the irrep of $G$ constructed above. 
We note that for any automorphism $\phi$ of $G$, 
$\rho_1\circ\phi: G\rightarrow U(2)$ is also a $2$-dimensional irrep of $G$ with the same image. 
If $\phi$ is an \emph{outer} automorphism of $G$, meaning it does not preserve conjugacy classes, then $\rho_1\circ\phi$ may be inequivalent to $\rho_1$. 
We construct such an outer automorphism: for $A\in SL(2,\mathbb{F}_5)$ (now viewed literally as a $2\times 2$ matrix over $\mathbb{F}_5$, not in any representation), consider
\begin{align}
    M = \left(\begin{array}{cc}
        1 &  0\\
        0 & 2
    \end{array}\right), \,\,\, \phi(A)=MAM^{-1}.
\end{align}
The automorphism $\phi$ swaps the conjugacy classes 
\begin{align}
\left(\begin{array}{cc}
        1 &  1\\
        0 & 1
    \end{array}\right)  \leftrightarrow \left(\begin{array}{cc}
        1 &  2\\
        0 & 1
    \end{array}\right),   \,\,\, \left(\begin{array}{cc}
        -1 &  1\\
        0 & -1
    \end{array}\right)  \leftrightarrow \left(\begin{array}{cc}
        -1 &  2\\
        0 & -1
    \end{array}\right),
\end{align}
while fixing the others, 
making it an outer automorphism. 
One can verify using GAP \cite{GAP4} that the character of $\rho_1$ has $9$ distinct values on the $9$ conjugacy classes of $G$. Since $\phi$ nontrivially permutes the conjugacy classes, the characters of the irreps $\rho_1$ and $\rho_2 = \rho_1\circ\phi$ must be distinct. Then $\rho_1$ and $\rho_2$ are inequivalent $2$-dimensional irreps. 

\begin{figure}[!ht]
    \centering
    \includegraphics[width=0.5\linewidth]{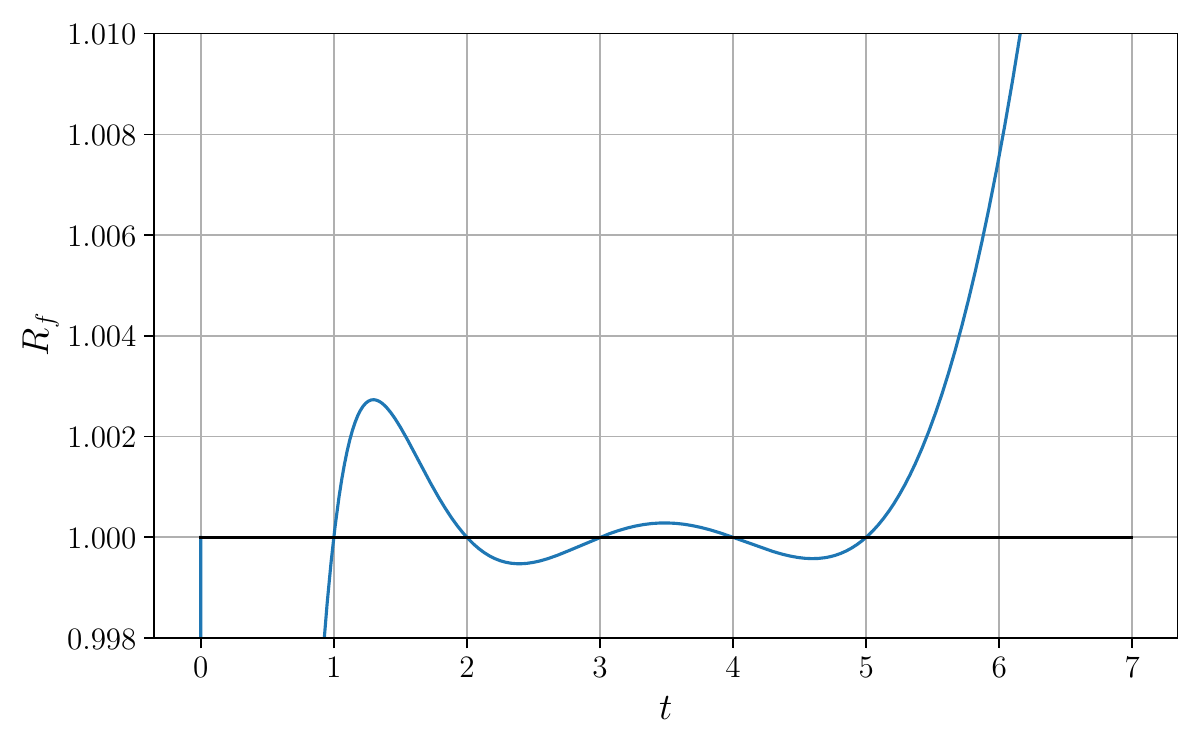}
    \caption{Frame potential ratio for a qubit (two-dimensional) representation of $SL(2,\mathbb{F}_5)$, which is a 5-design.}
    \label{fig:5-design}
\end{figure}

\subsubsection{Cyclic group frame potential}
In contrast, in \cref{fig:cyclic-fractional} we plot the frame potential ratio for small size cyclic groups $C_n$, and we see that these form a $t<0.5$ design with respect to $U(n)$. Here the cyclic group is generated by the generalized Pauli-X matrices (shift matrices), which permute basis elements $\ket{i}\rightarrow \ket{(i+s)\mod d}$.
It's easy to verify that the fractional frame potential for this group is: $F(t,d)=d^{2t-1}$. Noting that $\Gamma(1/2+1)=\sqrt{\pi}/2<1 = F(1/2,d)$, the value $t^*$ for which $\Gamma(t^*+1) = F(t^*,d)$ accumulates to $1/2$ as $d\rightarrow \infty$.

Another way to think about this result is that the generator of the cyclic group is unitarily equivalent to the qudit Pauli operator \(Z_d\). The group generated by this is an abelian group with elements \(\langle \mathbb{I}, Z_d, Z_d^2, \cdots, Z_d^{d-1} \rangle\) (modulo phases). It is easy to see that the only element with non-zero character is the identity, namely, \(\chi(\mathbb{I})=d, \chi(g)=0 ~\forall g \neq \mathbb{I}\). The character version of the fractional \(t\)-design test then yields,
\begin{align}
\frac{1}{|G|} \sum_{g \in G}|\chi(g)|^{2t} = d^{2t-1},
\end{align}
in any dimension \(d \in \mathbb{N}\) and \(t \in \mathbb{R}\).

\begin{figure}[!ht]
    \centering
    \includegraphics[width=0.5\linewidth]{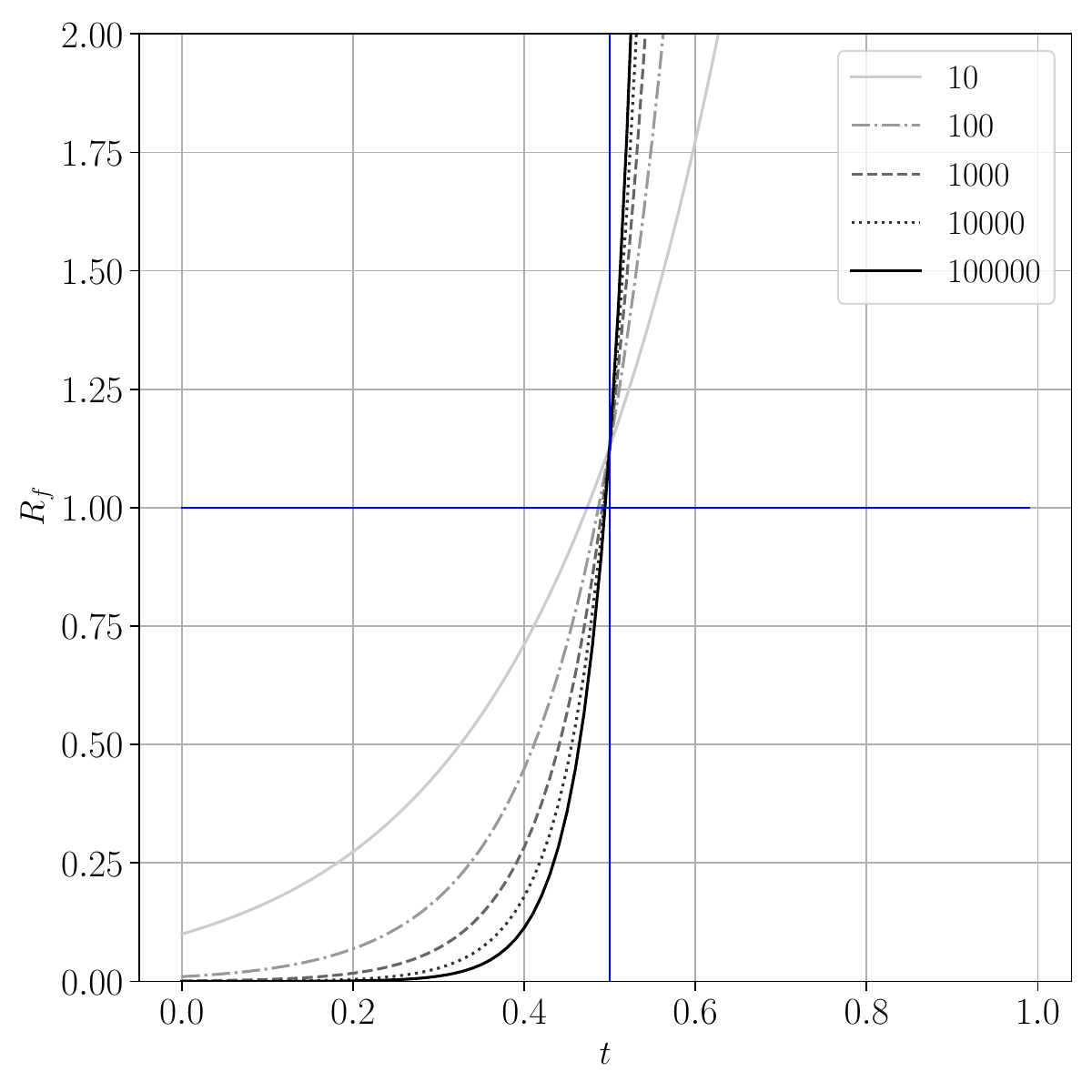}
    \caption{Fractional frame potential ratio for cyclic groups $C_d$ of order $d$ (see legend), with respect to $U(d)$. One can compute $R_f(d,t) = d^{2t-1}/\Gamma(t+1)$. For $d\rightarrow \infty$, the convergence point $t^*\rightarrow 1/2$ (see main text).}
    \label{fig:cyclic-fractional}
\end{figure}

\section{Closing}
\label{sec:closing}

\subsection{Applications}
\prlsection{Classical shadow tomography for qudits} Learning quantum systems is a computationally nontrivial problem. In the worst case, this requires exponential resources, for e.g., via quantum state and process tomography. However, in many physical scenarios, one is only interested not in learning the full quantum state, rather, the expectation value of this state over a set of observables (k-local operators, order parameters, fidelity, etc.). This has motivated the idea of classical shadow tomography \cite{Aaronson2020,huang_predicting_2020}. The key insight is that if we only want to learn a subclass of observables then one can circumvent the exponential bottleneck of full tomography. In particular, two types of shadow tomography protocols are particularly useful, (i) Pauli shadows and (ii) Clifford shadows. The two are so-called since they require averaging over the Pauli and the Clifford group of unitaries, respectively. Pauli shadows are useful for predicting low-weight observables (operators with limited support) while Clifford shadows are useful for predicting low-rank observables such as fidelity with a pure quantum state. In the multiqubit case, these protocols are very well-motivated since the Pauli and the Clifford group form a 1- and a 3-design, respectively. However, we need to reexamine this for the qudit case.

Note that if we want to do classical shadow tomography for qudits then we only need state designs, where we can utilize the weighted 2- and 3-designs introduced in this work. The key observation here is that weighted complex projective \(2\)-designs form a informationally complete POVM and are optimal for classical shadow tomography \cite{Innocenti2023ShadowPOVM}. This still leaves open the problem of efficient classical simulation of the snapshot states for qudit systems. This will either require the use of magic in the case of prime dimensional qudits, see e.g., some recent results in Refs. \cite{chen_nonstabilizerness_2024, mao_magic_2024}, or the use of random circuit approaches, see the discussion in \cite{akhtar_dual-unitary_2025}.

\prlsection{Information scrambling in qudits} Strongly-interacting quantum systems can spread local information quickly through the nonlocal degrees of freedom of a system. This property has been termed information scrambling, and is often quantified via the out-of-time-ordered correlator (OTOC) \cite{ovchinnikov_quasiclassical_nodate,kitaev_simple_2015,maldacena_bound_2016,lashkari_towards_2013,roberts_diagnosing_2015,polchinski_spectrum_2016,mezei_entanglement_2017,roberts_chaos_2017,hayden_black_2007,shenker_black_2014,swingle_unscrambling_2018,xu_swingle_tutorial_2022}. A common method to quantify the degree of chaoticity of an ensemble of unitaries is to ask if it can imitate the four-point OTOC for Haar-random unitaries. It turns out that unitary \(2\)-designs are sufficient for spoofing Haar-random unitaries via OTOCs. However, in the qudit case, the absence of unitary 2-groups means that one cannot spoof such dynamics. Of course using random quantum circuits, one can generate approximate unitary designs of any order (given sufficient depth of the circuit). However, here we are interested in studying exact group designs.

In the qubit case, OTOCs cannot tell apart unitary 2-designs generated via say the Clifford group from truly Haar-random unitaries. This has inspired the use of the variance of the OTOC to detect magic. The variance is sensitive to a unitary 4-design, which necessary requires magic. In fact, in qubit systems, any group of unitaries that forms a unitary 4-design must be universal. However, in the non-prime-power-dimensional system, OTOCs already can distinguish Haar-random unitaries from 2-designs. It will be interesting to understand the utility of OTOCs to probe the design structure of subgroups of the unitary group.

\prlsection{Randomized benchmarking for qudits} The standard Clifford RB scheme can no longer generate a unitary \(2\)-design if the dimension is not a power of a prime. As a result, the RB decay signal will be a linear combination of multiple exponentials, which is not ideal from the point of view of reliably fitting the data. However, the framework of character RB \cite{claes_character_2021} and filtered RB \cite{Helsen2022} are still useful to benchmark groups of unitaries. In particular, the results in Ref. \cite{Helsen2022} tell us that for filtered RB, if the measurement POVM utilized in the postprocessing of the RB data forms a complex projective 3-design, then the sample complexity of the entire RB method (collecting experimental data and its postprocessing to extract the exponentials) can be achieved with a number of samples that is independent of the dimension of the underlying Hilbert space. This result holds for arbitrary benchmarking groups \(G\) and is a sufficient condition (but not necessary). Therefore, even for non-prime-power dimensional qudit RB, our weighted state \(3\)-design guarantees scalability of the RB experiments, which is crucial to its applicability.

Moreover, given the considerable limitations that group-designs have for qudit systems, perhaps the random circuit based designs that are introduced in Ref. \cite{Helsen2022} are perhaps most useful from this point of view. It will be interesting to contrast this with the Clifford character RB that we have introduced.

\prlsection{Quantum machine learning for qudits} A prominent result in quantum machine learning states that if the parametrized quantum circuit used in training is sampled from a unitary \(2\)-design then the cost function necessarily runs into barren plateaus \cite{mcclean_barren_2018}. This follows from a simple application of the \textit{concentration of measure} ideas, along with the observation that the variance of the cost function is a polynomial of degree two and hence a unitary 2-design is sufficient to imitate this behavior. However, for qudit systems that are not a prime-power, the Clifford group does not form a unitary 2-design. Can we use this to say something about the learning/trainability of quantum circuits?

\subsection{Open Questions}
There are many open problems here that we list for the reader.
\begin{enumerate}
\item Can we generalize the weighted state t-design construction in an efficient way to multiqudit systems?

\item Can we use weighted unitary designs for quantum error correction, such as those originating from twisted unitary \(t\)-designs \cite{kubischta_quantum_2024}? In general, finding relationships between the twisted unitary t-group constructions and weighted t-designs would be useful.

\item Can we formalize the connection between fractional designs and fractional calculus, see e.g., Ref. \cite{jones_riemann-liouville_1991}?

\item Can we use the projection trick to define weighted unitary designs (as opposed to state designs)? Naively, this results in isometries that generate a symmetric projector. Since isometries have a quantum channel completion, perhaps this can be used to make a connection with quantum channel designs?

\item It will be interesting to see if fractional t-designs with \(1<t<2\) have any advantage over state 2-designs, perhaps in sampling?

\item Can we formalize the elements of a minimal state 2-design from the perspective of spin systems, namely \(SU(2)\) covariant rotations and spin squeezing? Do these suffice to generate state 2-designs or does it necessitate higher-order spin interactions?
\end{enumerate}

\subsection{Discussion}
Unitary designs have emerged as one of the most versatile tools in quantum information theory in the last two decades. They play a crucial role in a variety of tasks such as twirling of noise, which allows us to easily estimate quantum error correction thresholds; randomized benchmarking, that allows us to estimate average gate fidelities in a low-cost and scalable way; as well as classical shadow tomography, that allows us to learn many physical properties of exponentially large quantum states from only a few measurements. In a completely different context, they also find utility as minimal models of information scrambling as well as proofs of exponential separations in oracular complexity between quantum and classical computations. A large part of this effort has been focused towards efficient construction of approximate unitary designs. Instead, here we chose to focus on exact unitary designs in single qudit systems to highlight the immediate challenges one faces as we go beyond the qubit picture. 

To circumvent the lack of unitary \(t\)-groups beyond prime-power dimensions, we introduced a systematic method to construct weighted qudit \(2\)- and \(3\)-designs. This allows us to generalize standard qubit-based primitives to qudit architectures with little-to-no overhead. This also highlights that the "extra room" provided by qudit devices is not without its own challenges. With the rise of qudit platforms such as high-spin nuclei, cavity-QED systems, as well as photonics-based, this is a timely problem.

Our work also highlights the necessity of \textit{spin squeezing} for state \(2\)-designs in spin qudits, in analogy to the necessity of entanglement for multiqubit state \(2\)-designs. This strengthens the connections between the physics of high-dimensional spins and continuous-variable optical systems. In particular, while it is well-known result that (infinite-dimensional) optical coherent states cannot form a \(2\)-design; we show that so is true for spin-coherent states. However, the nature of these failures are completely different: for spin-coherent states this is not due to the "vanishing of the centroid" or divergence of integrals, rather, simply due to the representation theory of \(SU(2)\). The analogy between spin and optical coherent states has been particularly useful for quantum error correction, such as the development of spin-GKP codes, and we believe our work will only push this further. In particular, the proof that spin-GKP codewords form a state 2-design, just like optical GKP states is another addition to this framework.

We also proved bounds on the length of circuits needed to generate \(t\)-designs for both spin-qudits and cavity-QED systems, using random sequences of SNAPs and displacements, the native gate set in these architectures. And finally, we introduced a notion of "fractional \(t\)-designs," inspired from fractional calculus, which provides an alternative quantification of \(\epsilon\)-approximate designs. One of our key contributions is the introduction of a Clifford character randomized benchmarking scheme that can be utilized in any qudit dimension. This is one of the few schemes where a quantum information primitive works for a qudit dimension that is not a prime-power. It will be interesting to utilize these ideas to inspire novel QEC and practical applications for qudit systems, as well as unveiling their rich mathematical structure.

\section*{Acknowledgments}
N.A. is a KBR employee working under the Prime Contract No. 80ARC020D0010 with the NASA Ames Research Center. J.M. is thankful for support from NASA Academic Mission Services, Contract No. NNA16BD14C. This work was partly supported by the U.S. Department of Energy, Office of Science, National Quantum Information Science Research Centers, Superconducting Quantum Materials and Systems Center (SQMS) under contract number DE-AC02-07CH11359 through NASA-DOE interagency agreement SAA2-403602. A.M. acknowledges support of an Australian Research Council Laureate Fellowship (project no. FL240100181). The authors are grateful for the collaborative agreement between NASA and CQC2T. The United States Government retains, and by accepting the article for publication, the publisher acknowledges that the United States Government retains, a nonexclusive, paid-up, irrevocable, worldwide license to publish or reproduce the published form of this work, or allow others to do so, for United States Government purposes.

\bibliography{zotero,refs}

\begin{center}
\noindent\rule[0.5ex]{0.1\linewidth}{0.3pt}\rule[0.5ex]{0.1\linewidth}{0.8pt}\rule[0.5ex]{0.3\linewidth}{1.5pt}\rule[0.5ex]{0.1\linewidth}{0.8pt}\rule[0.5ex]{0.1\linewidth}{0.3pt}
\end{center}

\clearpage
\appendix

\renewcommand{\thepage}{A\arabic{page}}
\setcounter{page}{1}
\renewcommand{\thesection}{A\arabic{section}}
\setcounter{section}{0}
\renewcommand{\thetable}{A\arabic{table}}
\setcounter{table}{0}
\renewcommand{\thefigure}{A\arabic{figure}}
\setcounter{figure}{0}
\renewcommand{\theequation}{A\arabic{equation}}
\setcounter{equation}{0}

\section{Probability density for $d=2$ Haar random unitary eigenvalues \label{sect:eigen-density}}
It is known that the eigenvalues ($e^{i\phi_0}, e^{i\phi_1}$) of a Haar random matrix in dimension 2 are distributed according to density $p(\phi_0, \phi_1) = \frac{1}{4\pi^2} (1-\cos(\phi_1 - \phi_0))$ \cite{tracy_introduction_1992}. We are interested in the density for the difference $\Delta=|\phi_1-\phi_1|$ however. This can be computed by integrating $p(\phi_0, \phi_1)$ over lines of constant $\Delta$, an example of which is shown in Fig.~\ref{fig:grid}. Since the density function $p(\phi_1, \phi_0) \propto (1-\cos \Delta)$ is constant on such lines, the problem is purely geometric. We can see in Fig.~\ref{fig:grid} the length each line is $\sqrt{2}(2\pi - \Delta)$, and so the probability density for the difference is: $p(\Delta) \propto (2\pi-\Delta)(1-\cos\Delta)$. Normalization gives: 
\begin{align}
    p(\Delta) = \frac{1}{\pi}\left(1 - \frac{\Delta}{2\pi}\right)(1-\cos \Delta).
\end{align}

\begin{figure}[hb]
    \centering
    \includegraphics[width=0.4\linewidth]{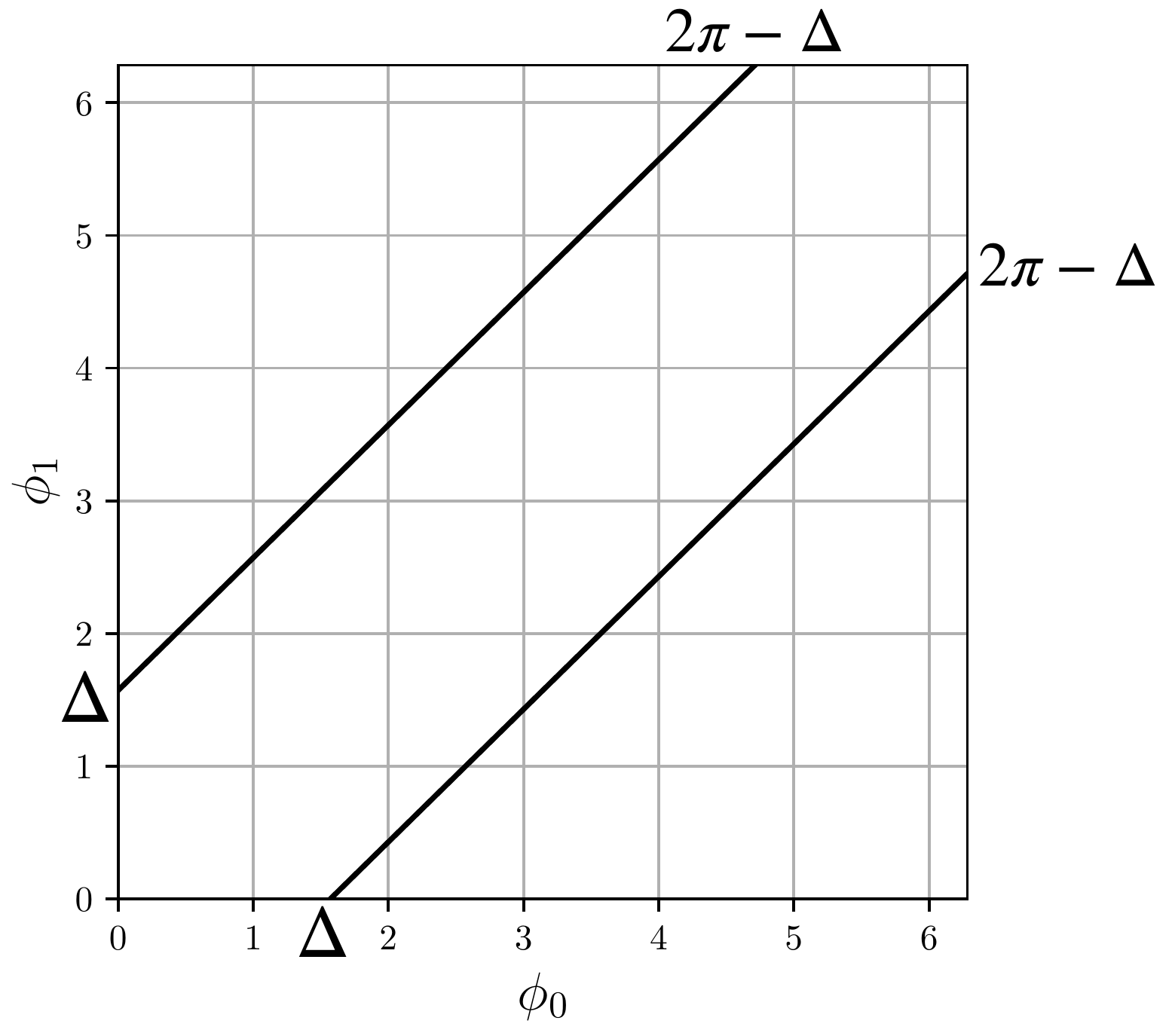}
    \caption{Lines of constant $|\phi_1-\phi_0|=\Delta$.}
    \label{fig:grid}
\end{figure}

\section{Evaluation of Eq.~\eqref{eq:d_2_integral}\label{sect:integral_proof}}
We wish to evaluate
\begin{align}
    I:=\frac{2^t}{\pi} \int_0^{2\pi} (1+\cos \Delta)^t (1-\cos \Delta)\left(1 - \frac{\Delta}{2\pi}\right) d\Delta.
\end{align}
In order to do so, we split it into two pieces:
\begin{align}
    I_0:= \pi \int_0^{2\pi} (1+\cos \Delta)^t (1-\cos \Delta) d\Delta, \;\; I_1:= \int_0^{2\pi} (1+\cos \Delta)^t (1-\cos \Delta)\Delta\, d\Delta,
\end{align}
where $I = \frac{2^t}{\pi^2}(I_0 - I_1/2)$.

We first observe that $I_0=I_1$, which we can see via:
\begin{align}
    I_1-I_0 = \int_0^{2\pi}(\Delta-\pi)(1+\cos \Delta)^t (1-\cos \Delta) d\Delta = \int_{-\pi}^\pi u (1 - \cos u)^t (1+ \cos u) du = 0.
\end{align}
The first step substitutes $u=\Delta-\pi$, and the last step follows by noting the integrand is odd, integrated over $[-\pi,\pi]$. Therefore, the integral of interest is $I = \frac{2^t}{2\pi^2}I_0$. 

We now focus on $I_0$, which we can rewrite using $1+\cos \Delta =2 \cos^2 (\Delta/2)$, $1-\cos \Delta = 2 \sin^2(\Delta/2)$:
\begin{align}
    I_0 = 2^{t+1}\pi \int_0^{2\pi} \left(\cos \frac{\Delta}{2}\right)^{2t}\left(\sin \frac{\Delta}{2}\right)^2 d\Delta = 2^{t+1}\pi \int_0^\pi (\cos u)^{2t} (\sin u)^2 du.
\end{align}
Noting that the integrand is symmetric about $\pi/2$, we get
\begin{align}
    I_0 = 2^{t+2}\pi \int_0^{\pi/2} (\cos u)^{2t} (\sin u)^2 du = 2^{t+1}\pi B\left(\frac{3}{2}, t + \frac{1}{2} \right),
\end{align}
where $B$ is the Beta function (valid for $\mathrm{Re}(t)>-1/2$).

Using standard identities for the Beta and Gamma functions, we get:
\begin{align}
    I_0 = 2^{t+1}\pi^{3/2} \frac{\Gamma(t+1/2)}{\Gamma(t+2)}.
\end{align}

This can be simplified to the final result using the \emph{Legendre duplication formula}:
\begin{align}
    I = \frac{2^t}{2\pi^2}2^{t+1}\pi^{3/2} \frac{\Gamma(t+1/2)}{\Gamma(t+2)}  = \frac{\Gamma(2t+1)}{\Gamma(t+1)\Gamma(t+2)}.
\end{align}

\section{Qubit unitary 5-design\label{sec:5-design-construction}}
In \cite{gross_evenly_2007} (Table III) and \cite{ketterer_entanglement_2020} (Appendix A) the details for a qubit unitary 5-design are given. However, in both references we believe there to be a mistake in the provided generators, and so, for completeness, we provide corrected versions here.

The group $G:=SL(2, \mathbb{F}_5)$ is the unitary group of $2\times 2$ determinant 1 matrices over $\mathbb{F}_5$. The group order is $|G|=120$, and it is isomorphic to the binary icosahedral ($2I$) group. First we start by providing corrected generators from the above papers (apparently a simple typo in each). We define $\omega = e^{2\pi i / 15}$.

First, it is apparent there is a mistake in the generators provided in \cite{gross_evenly_2007} (Table III), since one of the matrices does not have a determinant of absolute value one. A correct set of generators is:
\begin{align}
        & \left(\begin{array}{cc}
            -1 & 0 \\
            0 & -1
        \end{array} \right),\;\;
        \left(\begin{array}{cc}
            -\omega^{11}-\omega^{14} & -\omega^{11}-\omega^{14} \\
            \omega^{10} & -\omega - \omega^4
        \end{array} \right),\;\;\\
        & \left(\begin{array}{cc}
            -\omega-\omega^2-\omega^4-\omega^8 - 2(\omega^{11} + \omega^{14}) & \omega^{6} + \omega^{9} \\
            \omega^{11}-\omega^{14} & -\omega^{5}
        \end{array} \right).
\end{align}
These are identical to those in \cite{gross_evenly_2007}, apart from the third matrix, in the top left entry, where the sign in front of $\omega^{14}$ is flipped. These all have determinant one, and it can be checked they generate a group of size 120. 
We also note these can be used to define a set of $2I$ generators. Label the final two matrices above $g_1, g_2$ (the first matrix listed is redundant). Then $s=g_1 g_2^3 g_1$, and $t=g_1^2g_2^2$ satisfy $(st)^2 = s^3 = t^5$, which are known $2I$ generator relations.

Similarly for \cite{ketterer_entanglement_2020} (Appendix A), a corrected set of generators is:
\begin{align}
        & \left(\begin{array}{cc}
            -1 & 0 \\
            0 & -1
        \end{array} \right),\;\;
        \left(\begin{array}{cc}
            -\omega^{11}-\omega^{14} & \omega^6 + \omega^9 \\
            -\omega-\omega^2-\omega^4-\omega^7-\omega^8-\omega^{13} & \omega^{11}+\omega^{14}
        \end{array} \right),\;\;\\
        & \left(\begin{array}{cc}
            \omega^{10} & \omega^{11} + \omega^{14} \\
            -\omega^{2}-\omega^8 & -\omega^{10}
        \end{array} \right),
        \left(\begin{array}{cc}
            0 & \omega^5 \\
            -\omega^{10} & -\omega^3-\omega^{12}
        \end{array} \right).
\end{align}
Compared to \cite{ketterer_entanglement_2020} the last generator differs in the final entry. Also note that all four of these are not required, and e.g. the final three suffice to generate the group.
One can also define a set of $2I$ generators based on these. Label the final two generators above $k_2, k_3$. Then $s=k_2^3 k_3^2, t=k_3^3$ satisfies $(st)^2 = s^3 = t^5$, as required.

The above generators can be used to construct the full group of size 120, $\{g_k\}_{k=1}^{120}$, by matrix multiplication. Whilst these matrices are not unitary, the \emph{unitarian trick} can be employed to find a similarity transform to obtain a unitary representation of $G$ \cite{ketterer_entanglement_2020}:
\begin{align}
    U_k = \sqrt{P} g_k \sqrt{P}^{-1}, \,\, P=\frac{1}{|G|}\sum_{k=1}^{|G|}g_k^\dag g_k = 3 \left(\begin{array}{cc}
       1  &  \frac{1}{2}(1 + i \sqrt{5/3})\\
       \frac{1}{2}(1 - i \sqrt{5/3})  & 1
    \end{array}\right).
\end{align}

The $P$ matrix obtained from the Ref.~\cite{gross_evenly_2007} generators is:
\begin{align}
    P = \frac{1}{2}\left(\begin{array}{cc}
        9 + 3\sqrt{5} & 6 + 3\sqrt{5} + i\sqrt{27+12\sqrt{5}} \\
        6 + 3\sqrt{5} - i\sqrt{27+12\sqrt{5}} & 9 + 3\sqrt{5}
    \end{array}\right).
\end{align}

\section{Character theory and Schur polynomials}\label{sec:character appendix}

In this section, we will give some tools for decomposing representations of $SU(m)$ and related groups into irreps. We will do this by heavily using the Lie group-Lie algebra correspondence. 
The condition we will need is that the complexified Lie algebra of our group $G$ is isomorphic to $\mathfrak{sl}(m, \mathbb{C})$ or $\mathfrak{gl}(m, \mathbb{C})$, where $m\leq d$. 
This covers various cases, including $G=SU(2)$ as in the case of spin-coherent states, $G=SU(1,1)$ as in $SU(1,1)$ interferometry, and $G=U(m)$ as in the case of the linear optical unitary transformations of $n$ photons in $m$ modes. 

We begin by recalling some representation theory for the Lie algebra $\mathfrak{g} = \mathfrak{gl}(m,\mathbb{C})$. This is a special case of the representation theory of reductive Lie algebras over $\mathbb{C}$, and all of this can be generalized to that setting. 
We will later discuss the $m=2$ case, corresponding to $SU(2)$, in more detail. 
We omit a detailed introduction to the representation theory of Lie algebras, referring the reader to the classic exposition \cite{humphreys2012introduction}. The notions of equivalence, irreducibility, and so on are similar to the group case discussed in Section~\ref{sec:rep theory}. 

The (finite-dimensional) irreps of $\mathfrak{g}$ (up to equivalence) are in bijective correspondence with weakly decreasing integer tuples $\mu\in \mathbb{Z}^m$, called their \emph{highest weight}. We call the irrep corresponding to $\mu$ by $V_\mu$. 
Finite-dimensional representations $W$ of $\mathfrak{g}$ may be assigned a \emph{formal character} $\chi(W)\in\mathbb{Z}[x_1^{\pm 1}, \dots, x_m^{\pm 1}]$, a polynomial in $x_1^{\pm 1}, \dots, x_m^{\pm 1}$ with integer coefficients, satisfying the following properties: 
\begin{align}\label{eq:character relations}
    \chi(W_1\oplus W_2) = \chi(W_1) + \chi(W_2), \,\,\, \chi(W_1\otimes W_2) = \chi(W_1)\chi(W_2).
\end{align}
We have $\chi(V_\mu) = s_\mu(x_1, \dots, x_m)$, the Schur polynomial in $m$ variables. 
Notably, $\{s_\mu(x_1, \dots, x_m): \mu_1\geq \mu_2\geq\dots\geq\mu_m, \mu_i\in\mathbb{Z}\}$ is a $\mathbb{Z}$-basis for the subspace of symmetric polynomials in $\mathbb{Z}[x_1^{\pm 1}, \dots, x_m^{\pm 1}]$. 
By \eqref{eq:character relations}, we see that decomposing a representation $W$ of $\mathfrak{g}$ into a direct sum of irreducibles is equivalent to decomposing $\chi(W)$ into a nonnegative integer linear combination of Schur polynomials $s_\mu(x_1, \dots, x_m)$. The coefficients of the $s_\mu$ are precisely the multiplicities of the irreps $V_\mu$. 

Note that for any $c\in\mathbb{Z}$ and weakly decreasing $\mu\in\mathbb{Z}^r$, Schur polynomials satisfy
\begin{align}
    s_{(\mu_1 + c, \dots, \mu_m + c)}(x_1, \dots, x_m) = (x_1\cdots x_m)^c s_\mu(x_1, \dots, x_m).
\end{align}
Then, up to adjusting by powers of $x_1\cdots x_m$, it suffices to consider Schur polynomials $s_\mu$ where $\mu$ is a weakly decreasing $m$-tuple of \emph{nonnegative} integers. In this case, we have $s_\mu(x_1, \dots, x_m)$ a genuine polynomial in the $x_i$ (no inverses required). 
This is the setting in which many results about Schur polynomials are stated, such as those of \cite{macdonald1998symmetric}. 
We note that the factors of $x_1\cdots x_m$ do not seriously affect any of the results. 
In fact, the representation theory of $\mathfrak{sl}(m,\mathbb{C})$ is determined only by the differences $\mu_i - \mu_{i+1}$; the remaining degree of freedom only affects the scalar by which the identity matrix $I\in\mathfrak{g}$ acts, which has little effect on the representation theory. 

Now let $G\subseteq U(d)$ be a compact subgroup, and assume that the natural representation of $G$ on $\mathcal{V} = \mathbb{C}^d$ is irreducible. Further assume that the complexified Lie algebra of $G$ is isomorphic to $\mathfrak{sl}(m, \mathbb{C})$ or $\mathfrak{gl}(m, \mathbb{C})$, where $m\leq d$. 
In particular, assigning the central element $I\in\mathfrak{gl}(m, \mathbb{C})$ to act by an arbitrary integer if necessary, we obtain an irrep $(\rho,\mathcal{V})$ for $\mathfrak{g}=\mathfrak{gl}(m, \mathbb{C})$. 
Let $\lambda$ be the corresponding highest weight, so that we may identify $\mathcal{V} = V_\lambda$. 
For applications to randomized benchmarking (see Section~\ref{sec:character rb su2}), we will want to decompose the representation $V_\lambda\otimes V_\lambda^*$ into irreps. 
For this purpose, we note that $V_\lambda^*\cong V_{-w_0\lambda}$, where $w_0$ is the permutation that reverses the indices of $\lambda$ (so that $w_0 \lambda= (\lambda_m, \dots, \lambda_1)$). 
From \eqref{eq:character relations}, we may determine the relevant irreps and multiplicities by decomposing
\begin{align}
    s_\lambda(x_1, \dots, x_m) s_{-w_0\lambda}(x_1, \dots, x_m) = \sum_\nu c_\nu s_\nu(x_1, \dots, x_m).
\end{align}
These coefficients $c_\nu$ may be exactly calculated by the \emph{Littlewood-Richardson rule}. 
We will consider a special case: $\lambda = (k, 0, \dots, 0)$, where $k$ is a nonnegative integer. 
We have $s_\lambda = h_k$, the \emph{complete symmetric function}, and the Littlewood-Richardson rule simplifies to the \emph{Pieri rule}. (See \cite{macdonald1998symmetric}, Chapter I, (5.14).)
In particular, we have $-w_0\lambda = (0, \dots, 0, -k)$, and 
\begin{align}
    \chi(V_\lambda\otimes V_\lambda^*) &= \chi(V_\lambda)\chi(V_{-w_0\lambda})
    \\&=s_\lambda(x_1, \dots, x_m) s_{-w_0\lambda}(x_1, \dots, x_m) 
    \\&= (x_1\cdots x_m)^{-k} s_{(k,0, \dots, 0)}(x_1, \dots, x_m)s_{(k,\dots, k, 0)}(x_1, \dots, x_m) 
    \\&= (x_1\cdots x_m)^{-k} \sum_{\ell=0}^{k} s_{(2k-\ell, k, \dots, k, \ell)}(x_1, \dots, x_m)
    \\&= \sum_{\ell=0}^{k} s_{(k-\ell, 0, \dots, 0, \ell-k)}(x_1, \dots, x_m)
    \\&= \chi\left(\bigoplus_{\ell=0}^k V_{(k-\ell, 0, \dots, 0, \ell-k)} \right).
\end{align}
In particular, we have: 
\begin{proposition}\label{prop:multiplicity 1}
    Let $G\subseteq U(d)$ be a compact subgroup with complexified Lie algebra isomorphic to $\mathfrak{sl}(m, \mathbb{C})$ or $\mathfrak{gl}(m, \mathbb{C})$, where $m\leq d$. 
    Assume $\mathcal{V} = \mathbb{C}^d$ is an irrep of $G$ with highest weight $\lambda = (k, 0, \dots, 0)$. 
    Then $\textnormal{End}(\mathcal{V}) = V\otimes V^*$ decomposes into $k+1$ inequivalent irreps for $G$, each of multiplicity $1$. 
\end{proposition}

We now consider the special case of spin-coherent states: take $S$ a nonnegative half-integer and recall 
\begin{align}
    \mathcal{H}_S = \mathrm{span} \{ | S,k \rangle ~|~ k = -S, \cdots, S \}.
\end{align}
As discussed in Section~\ref{subsec:spin-coherent}, $\mathcal{H}_S$ is an irrep of $G=SU(2)$ with spin $S$. 
There is a correspondence between spin and highest weight; namely, an irrep of $\mathfrak{g} = \mathfrak{gl}(2, \mathbb{C})$ with highest weight $(a,b)$ has spin $S=(a-b)/2$ as an irrep of $G$. 
The irreps $V_{(a+c,b+c)}$ for $c\in\mathbb{Z}$ are all equivalent as $G$-representations, so we may choose the value of $c$ freely without affecting the results. 
Then we will view $\mathcal{H}_S$ as an irrep of $\mathfrak{g}$ with highest weight $(2S,0)$. 
From the above calculations, we see that
\begin{align}
     \mathcal{H}_S \otimes \mathcal{H}_S^* \cong \bigoplus_{\ell=0}^{2S}V_{(2S-\ell, \ell-2S)} \cong \bigoplus_{\ell=0}^{2S}\mathcal{H}_{2S-\ell}. 
\end{align}
In particular, there are $2S+1$ inequivalent irreps, all of multiplicity $1$. 

Note also that $\mathcal{H}_S$ and $\mathcal{H}_S^*$ are equivalent as representations of $SU(2)$, with isomorphism given by $\ket{S,k}\mapsto\bra{S,-k}$. 
Thus the above may be equivalently viewed as giving a decomposition of $\mathcal{H}_S\otimes \mathcal{H}_S$ into irreps. In the general case, $m>2$, $\mathcal{H}_S\otimes \mathcal{H}_S$ and $\mathcal{H}_S\otimes \mathcal{H}_S^*$ are not necessarily isomorphic, but they will always decompose into the same number of irreps, counting multiplicity (see the proof of Theorem B.1 in \cite{Saied2024}).

\begin{center}
\noindent\rule[0.5ex]{0.1\linewidth}{0.3pt}\rule[0.5ex]{0.1\linewidth}{0.8pt}\rule[0.5ex]{0.3\linewidth}{1.5pt}\rule[0.5ex]{0.1\linewidth}{0.8pt}\rule[0.5ex]{0.1\linewidth}{0.3pt}
\end{center}

\end{document}